\documentclass[12pt]{article}
\usepackage[utf8]{inputenc}
\usepackage[framemethod=TikZ]{mdframed}
\usepackage[margin=1.2in]{geometry}
\usepackage{soul}
\usepackage{authblk}
\usepackage{mathtools}
\usepackage{esint}
\usepackage[normalem]{ulem}

\usepackage{caption}
\usepackage{subcaption}
\usepackage{amsfonts}
\usepackage{lineno,hyperref}
\usepackage{bbm}
\usepackage{amsmath}
\usepackage{mathrsfs}
\usepackage{amssymb}
\usepackage{appendix}
\usepackage{amsmath,color}
\usepackage{amsthm}
\usepackage{amstext}
\usepackage{amsopn}
\usepackage{texdraw}
\usepackage{graphicx}
\usepackage{multirow}
\usepackage{epsfig,color}
\usepackage{lscape}
\usepackage{rotating}
\usepackage{float}
\usepackage{srcltx}
\usepackage[square, comma, sort&compress, numbers]{natbib}
\usepackage{epstopdf}
\newtheorem{theorem}{Theorem}[section]
\newtheorem{lemma}[theorem]{Lemma}

\newtheorem{proposition}[theorem]{Proposition}
\newtheorem{corollary}[theorem]{Corollary}

\theoremstyle{definition}
\newtheorem{definition}[theorem]{Definition}

\theoremstyle{remark}

\begin{document}

\date{}
\title{ \bf\large{Niche differentiation in the light spectrum promotes coexistence of phytoplankton species: a spatial modelling approach}\footnote{Partially supported by a US-NSF grant, an NSERC discovery grant RGPIN-2020-03911 and an NSERC discovery accelerator supplement award RGPAS-2020-00090.}}
\author{Christopher M. Heggerud\textsuperscript{1},\ \ King-Yeung Lam\textsuperscript{2},\ \ Hao Wang\textsuperscript{3}
 \\
{\small \textsuperscript{1} Department of Mathematical and Statistical Sciences, University of Alberta, Edmonton, AB, Canada. \hfill{\ }}\\
\ \ {\small Email: cheggeru@ualberta.ca\hfill {\ }}\\
{\small \textsuperscript{2} Department of Mathematics, Ohio State University, Columbus, OH, United States. \hfill{\ }}\\
\ \ {\small Email: lam.184@math.ohio-state.edu\hfill {\ }}\\
{\small \textsuperscript{3} Department of Mathematical and Statistical Sciences, University of Alberta, Edmonton, AB, Canada. \hfill{\ }}\\
\ \ {\small Email: hao8@ualberta.ca\hfill {\ }}}

\maketitle

%
%
%
%
\begin{abstract}
The paradox of the plankton highlights the apparent contradiction between Gause's law of competitive exclusion and the observed diversity of phytoplankton.  It is well known that phytoplankton dynamics depend heavily on two main resources: light and nutrients. Here we treat light as a continuum of resources rather than a single resource by considering the visible light spectrum. We propose a spatially explicit reaction-diffusion-advection model to explore under what circumstance coexistence is possible from mathematical and biological perspectives. Furthermore, we provide biological context as to when coexistence is expected based on the degree of niche differentiation within the light spectrum and overall turbidity of the water. 
\end{abstract}

\section{Introduction}

Phytoplankton are microscopic photosynthetic aquatic organisms that are the main primary producers of many aquatic ecosystems, and play a pivotal role at the base of the food chain. However, the overabundance of phytoplankton species, or algal blooms as it is often referred to, regularly leads to adverse effects both environmentally and economically~\cite{Huisman2018,Reynolds2006,Watson2015}. For these reasons the study of phytoplankton dynamics is important to enhance the positive effects of phytoplankton while limiting any adverse outcomes. Phytoplankton dynamics depend on inorganic materials, dissolved nutrients and light, creating energy for the entire aquatic ecosystem via photosynthesis~\cite{Reynolds2006}. As the world becomes more industrialized anthropogenic sources of nutrients drastically increase, more often than not, eutrophication ensues. Eutrophication is defined as the excess amount of nutrients in a system required for life. Thus, in eutrophic conditions light becomes the limiting resource for phytoplankton productivity~\cite{Paerl2013,Watson2015}. 

Resource limitation, be it light or nutrient limitation, leads to competition amongst species. By the competitive exclusion principle, any two species competing for a single resource can not stably coexist. However, several cases exist in nature that seemingly contradict the competitive exclusion principle such as Darwin's finches, North American Warblers, Anoles, as well as phytoplankton. However, these contradictions are easily explained through niche differentiation.  The paradox of the plankton (a.k.a Hutchinson paradox) stems from ostensible contradiction between the diversity of phytoplankton typically observed in a water body and the competitive exclusion principle, since phytoplankton superficially compete for the same resources~\cite{Hutchinson1961}.
Several modelling attempts have been made to shed light on this paradox by considering spatial heterogeneity throughout the water column~\cite{Jiang2019,Jiang2021,HsuLou2010}. For example, competitive advantage is gained by a species who has better overall access to light, whether it be realized through buoyancy regulation or increased turbulent diffusion. 

Classically, light has been treated as a single resource and competitive exclusion is regularly predicted by mathematical models~\cite{Heggerud2020,Huisman1994,Wang2007,Jiang2019,Jiang2021,HsuLou2010}. However, further investigation shows that phytoplankton species can absorb and utilize wavelengths with varying efficiencies, implying non-uniform absorption spectra~\cite{Burson2018,Luimstra2020,Holtrop2021,Stomp2007a}. A species' absorption spectrum measures the amount of light absorbed, of a specific wavelength, by the species. Figure \ref{fig:absspectra} gives examples of absorption spectra for four different species of phytoplankton. These differences between the absorption spectra imply niche differentiation among species and can, in part, help to explain Hutchinson's paradox.

Light limitation in aquatic systems occurs through several different mechanisms. For instance, incident light can be variable due to atmospheric attenuation, Rayleigh scattering and the solar incidence angle. All of these factors contribute to the amount of light that enters the water column. Moreover, light is attenuated by molecules and organisms as it penetrates through the vertical water column. Typically this attenuation is modelled using Lambert-Beer's law which assumes an exponential form of light absorbance by water molecules and seston (suspended organisms, minerals, compounds, gilvin, tripton and etc.). However, the amount of light attenuated is not strictly uniform with respect to wavelengths. For example, pure water absorbs green and red wavelengths more than blue, giving water its typical bluish tone whereas waters rich in gilvin, that absorb blue light, typically appear yellow. Additionally, as mentioned, phytoplankton species' absorption spectra are non-uniform across the light spectrum thus contributing to the variable light attenuation. Because absorption depends on wavelength, the available light profile can change drastically throughout the depth of the water column, giving rise to water colour and another mechanism for species persistence.
For these reasons, modelling of phytoplankton dynamics should explicitly consider light and its availability throughout the water column.

Several attempts have been made to study phytoplankton competition and dynamics. Single species models have been well established and give good understanding of the governing dynamics of phytoplankton in general~\cite{Heggerud2020,HsuLou2010,Du2011,Shigesada1981}. These studies include various modelling approaches including stoichiometric modelling~\cite{Heggerud2020}, non-local reaction-diffusion equations~\cite{HsuLou2010,Du2011} and complex limnological interactions~\cite{Zhang2021}. Non-local reaction diffusion equations are beneficial to the study of phytoplankton population because they are capable of capturing light availability after attenuation throughout the water column,  modelling diffusion and buoyancy/sinking of phytoplankton, and there exists a myriad of mathematical tools and theories to aid in their analysis. One such mathematical theory that we utilize in this paper is the monotone dynamical systems theory popularized by Smith~\cite{Smith}. The theory of monotone dynamical systems is a powerful tool to study the global dynamics of a complex competition system as utilized in~\cite{Jiang2019,Jiang2021,HsuLou2010}.

In this paper we extend spatially explicit mathematical models for phytoplankton dynamics to consider competition amongst phytoplankton species with niche differentiation in the absorption spectrum~\cite{Jiang2019,Jiang2021,Stomp2007}. Furthermore, the underwater light spectrum, and its attenuation, modelled by the Lambert-Beer law,  explicitly depends on the wavelengths of light. In Section \ref{sec:model}, we propose a reaction-diffusion-advection phytoplankton competition model that non-locally depends on phytoplankton abundance and light attenuation. In Section \ref{sec:mathresults}, we provide several preliminary results regarding the persistence of a single species via the associated linearized eigenvalue problem. In Section \ref{sec:Mathniche}, we introduce an index to serve as a proxy for the level of niche differentiation amongst two species and provide coexistence results based on this index. In the absence of niche differentiation we establish the competitive exclusion results based on advantages gained through buoyancy or diffusion. In Section \ref{sec:NumCoexistmechs}, we numerically explore how niche differentiation via i) specialist versus specialist competition, and ii) specialist versus generalist competition, can overcome competitive advantages that would otherwise result in competitive exclusion. We then consider the case when more than two species compete and show that upon sufficient niche differentiation any number of phytoplankton species may coexist in Section \ref{sec:Nspecies}. Finally, we offer a realistic competition scenario where the absorption spectra of two competing species are given in Figure \ref{fig:absspectra} and background attenuation is modelled based on water conditions ranging from clear to highly turbid. Our work offers a possible explanation of Hutchinson's paradox. That is, through sufficient niche differentiation in the light spectrum, many phytoplankton species can coexist.

\begin{figure}
    \centering
    \includegraphics[width=0.65\paperwidth]{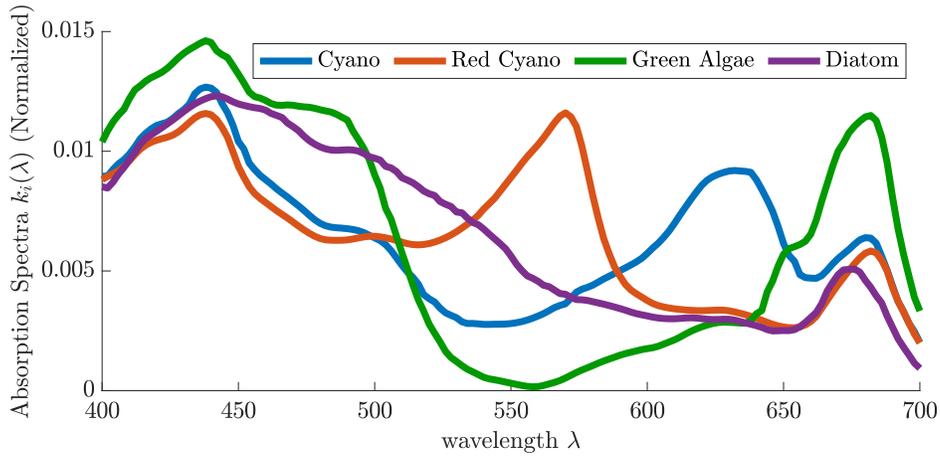}
    \caption{Normalized absorption spectra for four phytoplankton species: green cyanobacteria (\textit{Synechocystis} strain), red cyanobacteria (\textit{Synechococcus} strain), green algae (\textit{Chlorella} strain) and a diatom (\textit{Nitzschia} strain)~\cite{Luimstra2020,Burson2018,Stomp2007}. The differences of absorption spectra among species imply niche differentiation throughout the spectrum.}
    \label{fig:absspectra}
\end{figure}

\section{The model}\label{sec:model}
In this section we extend a two species non-local reaction-diffusion-advection model proposed in several papers~\cite{Jiang2019,Jiang2021,HsuLou2010,Du2010} to consider niche differentiation via absorption spectra separation. The PDE system assumes sufficient nutrient conditions so that light is the only factor limiting phytoplankton growth. However, the species are capable of utilizing incident wavelengths at varying efficiency as highlighted in Figure \ref{fig:absspectra}~\cite{Stomp2007,Burson2018,Luimstra2020,Holtrop2021}. Because of the attenuation of light through the vertical water column, the diffusivity of the phytoplankton and the potential for buoyancy regulation (advection) the system is spatially explicit. That is, let $x$ denote the vertical depth within the water column then $u_1(x,t)$ and $u_2(x,t)$ are the population density of competing phytoplankton species 1 and 2 at depth $x$ and time $t$. 
The following model generalizes the one of Stomp et al. \cite{Stomp2007} 
 to the spatial context:
\begin{equation}\label{eq:full}
\begin{cases}
\partial_t u_1 = D_1 \partial^2_xu_1 - \alpha_1 \partial_x u_1+ [g_1(\gamma_1(x,t)) - d_1(x)]u_1& \text{ for } 0 < x < L,\, t>0,\\
\partial_t u_2= D_2 \partial^2_xu_2- \alpha_2 \partial_x u_2+ [g_2(\gamma_2(x,t)) - d_2(x)]u_2& \text{ for } 0 < x < L,\, t>0,\\
D_1\partial_x u_1(x,t) - \alpha_1 u_1(x,t) = D_2\partial_x u_2(x,t) - \alpha_2 u_2(x,t)=0 & \text{ for }x = 0, L,\, t>0,\\
u_1(x,0) = u_{1,0}(x),\, u_2(x,0) =u_{2,0}(x)&\text{ for }0 < x < L.
\end{cases}
\end{equation}
In this paper,
the turbulent diffusion coefficients $D_1,D_2>0$ and sinking/buoyancy coefficients $\alpha_1,\alpha_2 \in \mathbb{R}$ are assumed to be constants; the functions $d_1(x),\,d_2(x)\in C([0,L] )$ are the death rate of the species at depth $x$; the function 
$\gamma_1(x,t)$ is the number of absorbed photons available for photosynthesis by species $1$ and is given by
\begin{equation}\label{eq:gamma1}
\gamma_1(x,t) = \int_{400}^{700}a_1(\lambda) k_1(\lambda) I(\lambda,x)\,d\lambda,
\end{equation}
where $k_1(\lambda)$ and $k_2(\lambda)$ are the absorption spectra of species $1$ and $2$, respectively. The absorption spectrum is the proportion of incident photons of a given wavelength absorbed by the cell. The respective quantity $\gamma_2(x,t)$ for species $2$ is similarly defined. 
For each given wavelength $\lambda$, the quantities
$a_1(\lambda)$ (resp. $a_2(\lambda)$) converts the absorption spectrum of species $1$, (and $u_2$) into the action spectrum, or the proportion of absorbed photons used for photosynthesis, of phytoplankton species $1$ (resp. $2$). In many cases photons are absorbed and utilized with similar efficiency, thus we take $a_1(\lambda) = a_2(\lambda)=1$.  Sunlight enters the water column with an incident light spectrum $I_{\rm in}(\lambda)$ and $I(\lambda,x,t)$ is the light intensity of wavelength $\lambda$ at depth $x$ which, according to the Lambert-Beer's law, given by
\begin{equation}
  I(\lambda,x,t) = I_{\rm in}(\lambda) \exp\left[- K_{BG}(\lambda)x - k_1(\lambda)\int_0^x u_1(y,t)-k_2(\lambda)\int_0^x u_2(y,t) \right], 
\end{equation}
where $K_{BG}(\lambda)$ is the background attenuation of the incident light spectrum. 
We also assume that the specific growth rates $g_1(s)$ and $g_2(s)$ of both phytoplankton species are increasing and saturating functions of the number of absorbed photons available for photosynthesis, i.e. 
\begin{equation}\label{eq:ggg}
g_i(0) = 0,\quad g'_i(s) >0 \quad \text{ for }s \geq 0, \quad g_i(+\infty)<+\infty \quad \text{ for }i=1,2. 
\end{equation}
A common choice of growth function is the Monod equation given by 
\begin{equation}
    g_i(s) = \frac{\bar{g}_i s}{\overline\gamma_i + s}, \quad \text{ i=1,2},
\end{equation}where $\bar{g}_{i}$ is the maximal growth rate of species $i$ and $\overline\gamma_i$ is the half-saturation coefficient.  
Lastly, we assume there is no net movement across the upper and lower boundaries of the water column, resulting in the zero-flux boundary conditions for $x=0, L$. 

\section{Preliminary results}\label{sec:mathresults}
In this section we establish several preliminary theorems for coexistence and competitive exclusion that are used throughout the paper. From the eigenvalue we establish conditions for a single species to persist in absence of a competitor. From this we are able to use the associated linearized eigenvalue problem to establish a sufficient condition for coexistence and competitive exclusion. 

Throughout the paper we refer the readers to the following definition and condition.
    Define the functions $f_i:[0,L] \times[0,\infty)\times [0,\infty)  \to \mathbb{R}$ by:
\begin{equation}\label{eq:f}
f_i(x,p_1,p_2) = g_i\left(\int_{400}^{700} a_i(\lambda) k_i(\lambda) I_{\rm in}(\lambda) \exp\bigg[ -K_{BG}(\lambda)x - \sum_{j=1}^2 k_j(\lambda) p_j\bigg]\right)  -d_i(x).
\end{equation} 
Then it is not hard to verify that, for $i=1,2$, the function $f_i$ satisfies
\begin{description}
\item[(H)] \quad $\displaystyle \frac{\partial f_i}{\partial p_j}<0$ \quad and\quad $\displaystyle \frac{\partial f_i}{\partial x}<0$ \quad for $(x,p_1,p_2) \in [0,L]\times \mathbb{R}_+^2$, \,$j=1,2$.
\end{description}
 Although we only consider the autonomous system here, we remark that most of the theoretical results can be generalized to the case of a temporally periodic environment.

\subsection{Persistence of a single species}\label{sec:SingleSpecies}

In this subsection we characterize the long-term dynamics of system \eqref{eq:full} in the absence of competition, i.e when  $u_{1,0} \equiv 0$ or $u_{2,0} \equiv 0$. We begin by defining the following eigenvalue problem.

\begin{definition}\label{def:eval}
For given constants $D>0$ and $\alpha \in \mathbb{R}$, and given function $h(x) \in C([0,L])$, define $\mu(D,\alpha,h) \in \mathbb{R}$ to be the smallest eigenvalue of the following boundary value problem:
\begin{equation}\label{eq:simpleeigenvalueproblem}
\begin{cases}
D \partial_{xx} \phi - \alpha \partial_x \phi + h(x)\phi  +\mu \phi =0 &\text{ for }(x,t) \in [0,L]\times\mathbb{R^+},\\
D\partial_x \phi - \alpha \phi = 0&\text{ for }(x,t) \in \{0,L\}\times \mathbb{R^+}.
\end{cases}    
\end{equation}
\end{definition}

 The eigenvalue problem given in Definition \ref{def:eval} is well associated to the system \eqref{eq:full} linearized around $E_1$ or $E_2$. The main result of this section is given below and provides a condition for the existence and attractiveness of the semi-trivial solutions $E_1$ and $E_2$.
\begin{proposition}\label{prop:semitrivial}
Suppose 
\begin{description}
\item[(P)] $\mu(D_i,\alpha_i,f_i(x,0,0)) <0$ for $i=1,2$.
\end{description}
Then the system \eqref{eq:full} has exactly two non-negative exclusion equilibria $E_1=(\tilde{u}_1,0)$ and $E_2=(0,\tilde{u}_2)$. 
Moreover, 
$E_1$ (resp. $E_2$) attracts all solutions of \eqref{eq:full} with initial condition $(u_{1,0},u_{2,0})$ such that 
$$
u_{1,0} \geq,\not\equiv 0\,\,\text{ and }\,\,u_{2,0}\equiv 0 \quad  {\rm(resp. }\quad 
u_{1,0} \equiv 0\,\,\text{ and }\,\,u_{2,0}\geq,\not\equiv 0).$$
\end{proposition}
\begin{proof}
See~\cite[Proposition 3.11]{Jiang2019}.
\end{proof}

 The following corollary gives an explicit condition for  {\bf(P)}.
\begin{corollary}\label{corr:Pholds}
Let $f_i$ be defined in \eqref{eq:f}. If 
\begin{equation}\label{eq:persist1}
 \int_0^L e^{\alpha_i x/D_i} f_i(x,0,0)\,dx >  0 \quad \text{ for }i=1,2,
\end{equation}
then {\bf(P)} holds and the conclusions of Proposition \ref{prop:semitrivial} concerning the existence and attractivity of semi-trivial solutions $E_1$ and $E_2$ hold.  
\end{corollary}
\begin{proof}
Thanks to Lemma \ref{lem:eigen} in the Appendix, \eqref{eq:persist1} implies {\rm{\bf(P)}}. The conclusion thus follows from Proposition \ref{prop:semitrivial}.
\end{proof}

In terms of the physical parameters, \eqref{eq:persist1} reads
\begin{equation}\label{eq:persist1'}   
\int_0^L e^{\alpha_i x/D_i}  g_i\left(\int_{400}^{700} a_i(\lambda) k_i(\lambda) I_{\rm in}(\lambda) e^{ -K_{BG}(\lambda)x}\right)\,dx >\int_0^L e^{\alpha_i x/D_i}   d_i(x) \,dx,
\end{equation}
giving an explicit condition for the existence and attractivity of the exclusion equilibrium $E_1$ and $E_2$.

\subsection{Coexistence in two species competition}\label{sec:twospecies}

We now consider the outcomes of a two species competition and establish sufficient conditions for coexistence. We begin by connecting the system \eqref{eq:full} to the general theory of monotone dynamical systems~\cite{Smith}. For this purpose, consider the cone $\mathcal{K} = \mathcal{K}_1 \times (- \mathcal{K}_1)$, where
\begin{equation}
    \mathcal{K}_1 = \left\{ \phi \in C([0,L])\,:\, \int_0^x \phi(y)\,dy \geq 0 \quad \text{ for all }x \in [0,L] \right\}.
\end{equation}
The cone $\mathcal{K}$ has non-empty interior, i.e. ${\rm Int}\,\mathcal{K}= ({\rm Int}\,\mathcal{K}_1)\times (-{\rm Int}\,\mathcal{K}_1)$, where
\begin{equation}\label{eq:interior}
   \rm Int\, \mathcal{K}_1 = \left\{ \phi \in C([0,L])\,:\,\phi(0)>0\, \int_0^x \phi(y)\,dy > 0 \quad \text{ for all }x \in [0,L] \right\}.
\end{equation}
For $i=1,2$,  let $(u_i(x,t),v_i(x,t))$ be two sets of solutions of \eqref{eq:full} with initial conditions $(u_{i,0}(x),v_{i,0}(x))$. Since $f_1$ and $f_2$ satisfy {\bf(H)}, it follows by~\cite[Corollary 3.4]{Jiang2019} that 
$$
(u_{2,0} - u_{1,0}, v_{2,0}-v_{1,0}) \in \mathcal{K} \setminus\{(0,0)\} \quad \Rightarrow \quad  (u_{2,0} - u_{1,0}, v_{2,0}-v_{1,0})(\cdot,t) \in {\rm Int}\,K \quad \forall t>0.
$$
In other words, the system \eqref{eq:full} generates a semiflow that is strongly monotone with respect to the cone $\mathcal{K}$. It follows from the property of monotone dynamical systems that the long-time dynamics of the system \eqref{eq:full} can largely be determined by the local stability of the equilibria.

We now characterize the local stability of $E_1$. 
\begin{proposition}[{~\cite[Proposition 4.5]{Jiang2019}}] \label{prop:4.5}
Suppose the parameters are chosen such that {\bf(P)} holds, i.e. the two species system has two exclusion equilibria $E_1=(\tilde{u}_1,0)$ and $E_2 = (0,\tilde{u}_2)$. 
\begin{itemize}
    \item[{\rm(a)}]
    The equilibria $E_1$ is linearly stable (resp. linearly unstable) if $\mu_u >0$ (resp. $\mu_u <0$), where
\begin{equation}\label{eq:muu}
\mu_u:=\mu(D_2,\alpha_2,f_2(x,\int_0^x\tilde{u}_1(y)\,dy,0)).
\end{equation}
    \item[{\rm(b)}]     The equilibria $E_2$ is linearly stable (resp. linearly unstable) if $\mu_v >0$ (resp. $\mu_v <0$), where
\begin{equation}\label{eq:muv}
\mu_v:=\mu(D_1,\alpha_1,f_1(x,\int_0^x\tilde{u}_2(y)\,dy,0)).
\end{equation}
\end{itemize}

\end{proposition}
\begin{proof}
We only prove assertion (a), since assertion (b) follows by a similar argument.
To determine the local stability of the exclusion equilibrium $E_1$, we consider
the associated linearized eigenvalue problem at $E_1=(\tilde{u}_1,0)$, which is given by
\begin{equation} \label{eq:lineigproblem}
    \begin{cases}
     D_1 \phi_{xx} - \alpha_1 \phi_x + f_1(x,\int_0^x\tilde{u}_1(y)\,dy,0) \phi &\\ 
     \quad - \tilde{u}_1 g_1'(\gamma_1) 
     [A_{11}(x) \int_0^x \phi(y)\,dy + A_{12}(x) \int_0^x \psi(y)\,dy] + \mu \phi = 0 &\text{ in }[0,L],\\
    D_2 \psi_{xx} - \alpha_2 \psi_x + f_2(x,\int_0^x\tilde{u}_1(y)\,dy,0)\psi + \mu \psi = 0 &\text{ in }[0,L],\\
    D_1\phi_x - \alpha_1 \phi = D_2 \psi_x - \alpha_2 \psi = 0 & \text{ for }x = 0,L.
    \end{cases}
\end{equation}
where (recall that we have taken $a_i\equiv 1$) 
\begin{equation}
A_{ij}(x) = \int_{400}^{700} k_i(\lambda)I(\lambda,x)k_j(\lambda)\,d\lambda
\end{equation}
and 
\begin{equation}
\gamma_i(x)=\int_{400}^{700} k_i(\lambda)I_{in}(\lambda) \exp \left[ - K_{BG}(\lambda)x - k_1(\lambda) \int_0^x \tilde{u}_1(y)\,dy \right]\,d\lambda.
\end{equation}
Thanks to the monotonicity of the associated semiflow, 
the linearized problem \eqref{eq:lineigproblem} has a principal eigenvalue in the sense that $\mu_1 \leq \textup{Re}\,\mu$ for all eigenvalues $\mu$ of \eqref{eq:lineigproblem}, and that the corresponding eigenfunction can be chosen in $\mathcal{K} \setminus\{(0,0)\}$. In particular, $E_1$ is linearly stable (resp. linearly unstable) if $\mu_1 >0$ (resp. $\mu_1 <0$). 

Next, we apply~\cite[Proposition 4.5]{Jiang2019}, which says that 
$$
{\rm sgn}\, \mu_1 = {\rm sgn}\, \mu_u, 
$$
where $\mu_u$, given in \eqref{eq:muu}, is the principal eigenvalue of the second equation in \eqref{eq:lineigproblem}. 
Hence, $E_1$ is linearly stable (resp. linearly unstable) if $\mu_1 >0$ (resp. $\mu_1 <0$). 
\end{proof}

If both $E_1$ and $E_2$ exist we can conclude the existence of a positive equilibrium solution by the following proposition. 
\begin{proposition}
Assume {\bf(P)}, so that both semi-trivial equilibria $E_1$ and $E_2$ exist. Suppose further that
$$
\mu_u \cdot \mu_v >0,
$$
then \eqref{eq:full} has at least one positive equilibrium $(\hat{u}_1,\hat{u}_2)$. 
\end{proposition}
\begin{proof}
If $\mu_u \cdot \mu_v >0$, then the exclusion equilibria $E_1$ and $E_2$ are either both linearly stable or both linearly unstable. The existence of positive equilibrium thus follows from~\cite[Remark 33.2 and Theorem 35.1]{Hess}.
\end{proof}
In case both $E_1$ and $E_2$ are linearly unstable, both species persist in a robust manner. 
\begin{proposition}\label{prop:suffcoex}
Assume {\bf(P)} so that the semi-trivial equilibria $E_1,E_2$ exist. Suppose 
\begin{equation}\label{eq:suffcoexist}
\mu_u <0 \quad \text{ and }\quad \mu_v <0,
\end{equation}
(i.e. both $E_1$ and $E_2$ are unstable) then the following holds.
\begin{itemize}
    \item[{\rm(i)}] There exists $\delta_0>0$ that is independent of the initial data such that
    $$
    \liminf_{t \to \infty} \min_{i=1,2} \int_{0<x<L} u_i(x,t) \geq \delta_0;
    $$
    \item[{\rm(ii)}] System \eqref{eq:full} has at least one coexistence, equilibrium $(\hat{u}_1,\hat{u}_2)$ that is locally asymptotically stable.
\end{itemize}
\end{proposition}
\begin{proof}
By \eqref{eq:suffcoexist}, both exclusion equilibria $E_1, E_2$ are linearly unstable. The result follows from~\cite[Theorems 33.3]{Hess}.
\end{proof}

The signs of the principal eigenvalue $\mu_u$ and $\mu_v$ are often difficult to determine. We now establish an explicit condition for coexistence. To this end, observe from Corollary \ref{corr:Pholds}  and \eqref{eq:persist1'} that a sufficient condition for
\begin{equation}
    \mu_v= \mu(D_1,\alpha_1,f_1(x,0,\int_0^x \tilde{u}_{2}(y,t)\,dy))<0.
\end{equation}
is given by
\begin{multline}\label{eq:suffcoexist'}   
 \int_0^L e^{\alpha_1 x/D_1}  g_1\left(\int_{400}^{700} a_1(\lambda) k_1(\lambda) I_{\rm in}(\lambda) e^{ -K_{BG}(\lambda)x- k_2(\lambda) \int_0^x \tilde{u}_2(y,t)\,dy}\right)\,dx 
 \\>\int_0^L e^{\alpha_1 x/D_1}   d_1(x,t) \,dx.
\end{multline}

For $i=1,2$, we will obtain an explicit upper bound for $\int_0^x \tilde{u}_i(y)\,dy$. To this end, define 
$$
M_i:= \inf\left\{M >0:\,\, \int_0^xf_i(y,0, M\int_0^ye^{-\alpha_i z/D_i}\,dz)e^{-\alpha_i y/D_i}\,dy \leq 0 \text{ in }[0,L]\times[0,T]\right\}.  
$$
\begin{lemma}\label{lem:upperboundsM2}
For $i=1,2$, 
$$
\int_0^x\tilde{u}_i(y,t)\,dy \leq \frac{M_i D_i}{\alpha_i} (1-e^{-\alpha_i x/D_i}) 
\quad \text{ for all }(x,t) \in [0,L]\times[0,T]. 
$$
\end{lemma}
\begin{proof}
Indeed, with such a choice of $M_i$, the function $M_ie^{-\alpha_i x/ D_i}$ will then qualify as an super-solution for the single species equation for species $i$, in the sense of~\cite[Subsection 3.2]{Jiang2019}. Hence, by comparison, we have
$$
M_i e^{-\alpha_i x/D_1} - \tilde{u}_i \in \mathcal{K}_1, 
$$
that is, 
$$
\int_0^x \tilde{u}_i(y,t)\,dy \leq \int_0^x M_ie^{-\alpha_i y/ D_i}\,dy = \frac{M_i D_i}{\alpha_i} (1-e^{-\alpha_i x/D_i}) \quad \text{ for }x \in [0,L].
$$
This completes the proof.
\end{proof}
By the above discussion, a sufficient condition for \eqref{eq:suffcoexist'} is
\begin{multline}
    \int_0^L e^{\alpha_1 x/D_1}  g_1\left(\int_{400}^{700} a_1(\lambda) k_1(\lambda) I_{\rm in}(\lambda,t) e^{ -K_{BG}(\lambda)x- k_2(\lambda)  \frac{M_2 D_2}{\alpha_2} (1-e^{-\alpha_2 x/D_2})}\right)\,dx \\
 \qquad > \int_0^L e^{\alpha_1 x/D_1}   d_1(x) \,dx. \label{eq:suffcoexist''}  
\end{multline} 

Furthermore, an upper bound, $M_1$, for $\int_0^x\tilde{u}_1(y,t)\,dy$ is easily established following the arguments in Lemma \ref{lem:upperboundsM2}. Thus, a sufficient condition for \eqref{eq:suffcoexist} is given by \eqref{eq:suffcoexist''} and 
\begin{multline}
\int_0^L e^{\alpha_2 x/D_2}  g_2\left(\int_{400}^{700} a_2(\lambda) k_2(\lambda) I_{\rm in}(\lambda) e^{ -K_{BG}(\lambda)x- k_1(\lambda)  \frac{M_1 D_1}{\alpha_1} (1-e^{-\alpha_1 x/D_1})}\right)\,dx \\
\qquad > \int_0^L e^{\alpha_2 x/D_2}   d_2(x) \,dx. \label{eq:suffcoexistu}  
 \end{multline} 
This yields an explicit sufficient condition for coexistence.

\section{Extreme cases of niche differentiation: competitive outcomes}\label{sec:Mathniche}

In this section, we explicitly consider niche differentiation via the absorption spectra, $k_1(\lambda)$ and $k_2(\lambda)$. We consider the extreme cases of differentiation, where the niches either completely overlap, or do not overlap at all. Sufficient conditions for exclusion or coexistence are given. 

We establish the following definition to serve as a proxy for niche differentiation.

\begin{definition}

\begin{equation}\label{def:I_S}
    \mathcal{I}_S(k_1,k_2)= \frac{\|k_1-k_2\|_{L^1}}{\|k_1\|_{L^1} + \|k_2\|_{L^1}}. 
\end{equation}

\end{definition}
We refer to $\mathcal{I}_S(k_1,k_2)$ as the index of spectrum differentiation among two species. If the two species have the same absorption spectra then $\mathcal{I}_S(k_1,k_2)=0$ whereas if their absorption spectra are completely non-overlapping then $\mathcal{I}_S(k_1,k_2)=1$.

\subsection{Coexistence for disjoint niches}\label{subsect:exclude}
 Consider the case where the absorption spectra are completely non-overlapping, so that
competition for light is at the extreme minimum. Namely,  $$\mathcal{I}_S(k_1(\lambda),k_2(\lambda))=1.$$ 
We give a coexistence result.
\begin{corollary}\label{lem:disjoint}
Suppose {\bf(P)} holds, so that the exclusion equilibria $E_1$ and $E_2$ exist. If, in addition, $\mathcal{I}_S(k_1(\lambda),k_2(\lambda))=1$, then the coexistence results of Proposition \ref{prop:suffcoex} hold.
\end{corollary}
\begin{proof}
First note that $\mathcal{I}_S(k_1(\lambda),k_2(\lambda))=1$ is equivalent to $k_1(\lambda)k_2(\lambda)=0$ for each $\lambda$. 
It suffices to observe that 
$$
f_2(x,\int_0^x \tilde{u}_1(y)\,dy,0) = f_2(x,0,0),\quad \text{ and }\quad 
f_1(x,0,\int_0^x \tilde{u}_2(y)\,dy) = f_1(x,0,0)
$$
so that {\bf(P)} implies $\mu_u <0$ and $\mu_v<0$. The rest follows from Proposition \ref{prop:suffcoex}.
\end{proof}

\subsection{Competitive exclusion for identical niches}\label{sec:sameniche}
Next, we consider the case where the absorption spectra overlap completely ($\mathcal{I}_S(k_1,k_2)=0$) to consider maximum competition for light. Recall our assumption that $a_1(\lambda)=a_2(\lambda)=1$. Thus, Under these assumptions we establish the competitive exclusion scenarios in the following theorems.

\begin{theorem}{~\cite[Theorem 2.2]{Jiang2019}}
Assume $\mathcal{I}_S(k_1,k_2)=0$. Let $D_1 = D_2$, $\alpha_1 < \alpha_2$, $f_1 = f_2$, $d_1=d_2$. If ${\bf (P)}$ holds $($i.e. both $E_1,E_2$ exist$)$, then species $1$ drives the second species to extinction, regardless of initial condition.
\label{thm:4.3}
\end{theorem}
\begin{proof}
By the theory of monotone dynamical systems (see, e.g.~\cite[Theorem B]{Hsu1996competitive} and~\cite[Theorem 1.3]{Lam2016remark}), it suffices to establish the linear instability of the exclusion equilibria $E_2$, and the non-existence of positive equilibria.

\noindent {\bf Step 1.} We claim that $\mu_v <0$, i.e. $E_2=(0,\tilde{u}_2)$ is linearly unstable.

Recall that $\tilde{u}_2$ is the unique positive solution to
$$
\begin{cases}
D_2 \tilde u_{xx} - \alpha_2 \tilde u_x + f_2(x,0, \int_0^x \tilde u(y)\,dy) \tilde u = 0 &\text{ in }[0,L],\\
D_2 \tilde{u}_x - \alpha_2 \tilde u = 0 &\text{ for }x=0,L,
\end{cases}
$$
where $f_2$ is given in \eqref{eq:f} and satisfies {\bf(H)}. Since $\tilde{u}_2$ can be regarded as a positive eigenfunction, we deduce that $\mu(D_2,\alpha_2,f_2(x,0,\int_0^x \tilde{u}_2(y)\,dy)) =0$. 

Since $D_1=D_2$, $\alpha_1<\alpha_2$ and $f_1=f_2$, we may apply Lemma \ref{lem:4.3}(a), found in the Appendix, to get
$$
\mu_v=\mu(D_1,\alpha_1,f_1(x,0,\int_0^x \tilde{u}_2(y)\,dy)) <\mu(D_2,\alpha_2,f_2(x,0,\int_0^x \tilde{u}_2(y)\,dy)) = 0. 
$$
Thus $E_2$ is linearly unstable.

\noindent {\bf Step 2.} The system \eqref{eq:full} has no positive equilibrium.

Suppose to the contrary that $(u_1^*, v^2_*)$ is a positive equilibrium, then deduce that
$$
\mu(D_i,\alpha_i,f_i(x,\int_0^x u^*_1(y)\,dy, \int_0^x u^*_2(y)\,dy))=0 \quad \text{ for }i=1,2,
$$
where the respective eigenfunctions are given by $u^*_i>0$. However, this is in contradiction with Lemma \ref{lem:4.3}(a).
\end{proof}

\begin{theorem}{~\cite[Theorem 2.3]{Jiang2019}}
Assume $\mathcal{I}_S(k_1,k_2)=0$. Let $D_1 < D_2$, $\alpha_1 = \alpha_2 \geq [f_1(L,0,0) - d_1]L$, $f_1 = f_2$, $d_1=d_2$. If ${\bf (P)}$ holds $($i.e. both $E_1,E_2$ exist$)$, then the faster species , species $2$ drives the slower species, species $1$, to extinction, regardless of initial condition.
\label{thm:4.4}
\end{theorem}
\begin{proof}
Denote $\alpha = \alpha_1=\alpha_2$ and $f=f_1=f_2$.
By the theory of monotone dynamical systems (see, e.g.~\cite[Theorem B]{Hsu1996competitive} and~\cite[Theorem 1.3]{Lam2016remark}), it suffices to establish the linear instability of the exclusion equilibria $E_2$, and the non-existence of positive equilibria.

\noindent {\bf Step 1.} We claim that $\mu_u <0$, i.e. $E_1=(\tilde{u}_1)$ is linearly unstable.

Recall that $\tilde{u}_1$ is the unique positive solution to
$$
\begin{cases}
D_1 \tilde u_{xx} - \alpha \tilde u_x + f(x,0, \int_0^x \tilde u(y)\,dy) \tilde u = 0 &\text{ in }[0,L],\\
D_1 \tilde{u}_x - \alpha \tilde u = 0 &\text{ for }x=0,L,
\end{cases}
$$
where $f=f_1=f_2$ is given in \eqref{eq:f} and satisfies {\bf(H)}. Since $\tilde{u}_1$ can be regarded as a positive eigenfunction, we deduce that $\mu(D_1,\alpha,f(x,0,\int_0^x \tilde{u}_1(y)\,dy)) =0$. 

Next, we claim that 
\begin{equation}
\mu_u=\mu(D_2,\alpha,f(x,0,\int_0^x \tilde{u}_1(y)\,dy)) < 0. 
\end{equation}
Suppose to the contrary that $H(D_2) \geq 0$, where
$$
H(D) := \mu(D,\alpha_1,f_1(x,0,\int_0^x \tilde{u}_2(y)\,dy)).
$$
Since $D_1<D_2$, $\alpha \geq [f(0,0,0)]L$, we have $H(D_1)=0$ and $H(D_2)\geq 0$. By Lemma \ref{lem:4.3}(c) (found in the Appendix), $H'(D_1)<0$, so that there exists $D_3 \in (D_1,D_2]$ such that $H(D_3)=0$ and $H'(D_3) \geq 0$. But this is impossible in view of Lemma \ref{lem:4.3}(c).
Thus $E_1$ is linearly unstable.

\noindent {\bf Step 2.} The system \eqref{eq:full} has no positive equilibrium.

Suppose to the contrary that $(u_1^*, u^*_2)$ is a positive equilibrium, then deduce that
$$
\mu(D_i,\alpha,f(x,\int_0^x u^*_1(y)\,dy, \int_0^x u^*_2(y)\,dy))=0 \quad \text{ for }i=1,2,
$$
where the respective eigenfunctions are given by $u^*_i>0$. However, we can argue as in Step 1 that this is in contradiction with Lemma \ref{lem:4.3}(c).
\end{proof}
\begin{theorem}{~\cite[Theorem 2.4]{Jiang2019}}
Assume $\mathcal{I}_S(k_1,k_2)=0$. Let $D_1 < D_2$, $\alpha_1 = \alpha_2 \geq 0$, $f_1 = f_2$, $d_1=d_2$. If ${\bf (P)}$ holds $($i.e. both $E_1,E_2$ exist$)$, then the slower species, species $1$ drives the faster species, species $2$, to extinction, regardless of initial condition.
\label{thm:4.5}
\end{theorem}
\begin{proof}
The proof is the same as Theorem \ref{thm:4.3}, found in the Appendix, where we use Lemma \ref{lem:4.3}(b) instead of Lemma \ref{lem:4.3}(a).
\end{proof}
Note that $\mathcal{I}_S(k_1,k_2)=0$ is equivalent to $k_1(\lambda)=k_2(\lambda)$ for all $\lambda$. 
The above theorems can be summarized into a single sentence: Suppose both species consume light in the same efficiency, the species that remains at, or moves towards the water's surface at a higher rate will exclude the other species. That is, if both species are sinking either the one sinking slower, or with higher diffusion will exclude. If both species are buoyant then the less buoyant species or the more diffusive species will be excluded.

\section{Numerical investigation of niche differentiation}\label{sec:NumCoexistmechs} 
To complement the theorems established in Sections \ref{sec:mathresults} and \ref{sec:Mathniche} we present several numerical simulations that show the relatively large regions in parameter space that allow for coexistence. We numerically explore two main competition scenarios: 1) Niche differentiation through specialization of different wavelengths and 2) niche differentiation through specialist and generalist (with respect to light) competition. In each scenario we consider the intermediate levels of niche differentiation evaluated by $\mathcal{I}_S(k_1,k_2)$.

The main results of this section can be summarized by the following key points: 
\begin{itemize}
    \item[P1:] Competitive advantage is given to the species whose absorption spectrum overlaps the most with the available incident light. However, significant niche differentiation can promote coexistence for scenarios where incident light does not strongly favour a single species. 
    \item[P2:] Competitive exclusion through an advection advantage can be overcome by niche differentiation. 
    \item[P3] Intermediate values of specialization will promote coexistence. Otherwise, the specialist is excluded if its niche is too narrow, or excludes if its niche overlaps with the incident light significantly.  
\end{itemize}

\subsection{Competition outcomes for specialization on separate parts of the light spectrum}\label{sec:specvsspec}
Here we assume that the two species with relatively narrow niches are competing for light. We numerically show that through niche differentiation a species can resist competitive exclusion. These results imply that without the assumption of $\mathcal{I}_S(k_1,k_2)=0$ the theorems in Section \ref{subsect:exclude} do not hold and that when species' absorption spectra do not significantly overlap, coexistence is readily observed. 

To investigate the extent of which niche differentiation promotes coexistence we consider two scenarios. First, we let $k_1(\lambda)$ and $k_2(\lambda)$ be unimodal functions that are horizontal translations of each other. That is, let $g^*(\lambda)$ be a truncated Gaussian distribution on (-75,75) with mean zero and variance $\sigma$. Then $k_i(\lambda)=g^*(\lambda-\lambda_{i,0})$ where $\lambda_{i,0}\in[475,625]$ is the location of peak absorbance in the visible light spectrum. This ensures $k_1(\lambda)$ and $k_2(\lambda)$ have the same $L^1$ norm and are identical in their degree of specialization, giving no advantage through the absorption spectra alone.  We then allow the location of peaks of $k_2(\lambda)$ to vary along the light spectrum ($\lambda_{2,0}\in[475,625]$) while keeping $k_1(\lambda)$ fixed $\lambda_{1,0}=475$. By varying the location of the peak of $k_2(\lambda)$ we in-turn vary  $\mathcal{I}_S(k_1,k_2)$. Examples of this are shown graphically with the blue curves in Figure \ref{fig:explainspecvspec}. We also assume that the incident light $I_{in}(\lambda)$ is a unimodal function with the location of peak incidence at $\lambda=\lambda_{I}$. To understand the implications incident light has on coexistence we vary $\lambda_I$ in the range $[450,650]$). Two example curves for $I_{in}(\lambda)$ are shown in orange in Figure \ref{fig:explainspecvspec}. 

Second, we alter $\mathcal{I}_S(k_1,k_2)$ as above but with a uniform incident light function and allow a competitive advantage through advection by altering the advection rate $\alpha_2$ of species $2$. Recall that $u_1$ has competitive advantage when $\alpha_1<\alpha_2$, and species $2$ has competitive advantage when $\alpha_1>\alpha_2$ (see Theorem \ref{thm:4.3}).

By varying $\mathcal{I}_S(k_1,k_2)$ we can then explore the competitive outcomes for various scenarios where exclusion is known to occur when niche differentiation is not considered. Furthermore, we show that the incident light function $I_{in}(\lambda)$, together with the absorption spectra $k_1(\lambda),k_2(\lambda)$, play important roles in the competition outcome by allowing competitive advantages to be overcome, or diminished. Our results of this section are shown in Figure \ref{fig:specvsSpec} and \ref{fig:specvsspecAdvdist}.
\begin{figure}
\centering
\begin{subfigure}[b]{0.49\textwidth}
     \centering
    \includegraphics[width=\textwidth]{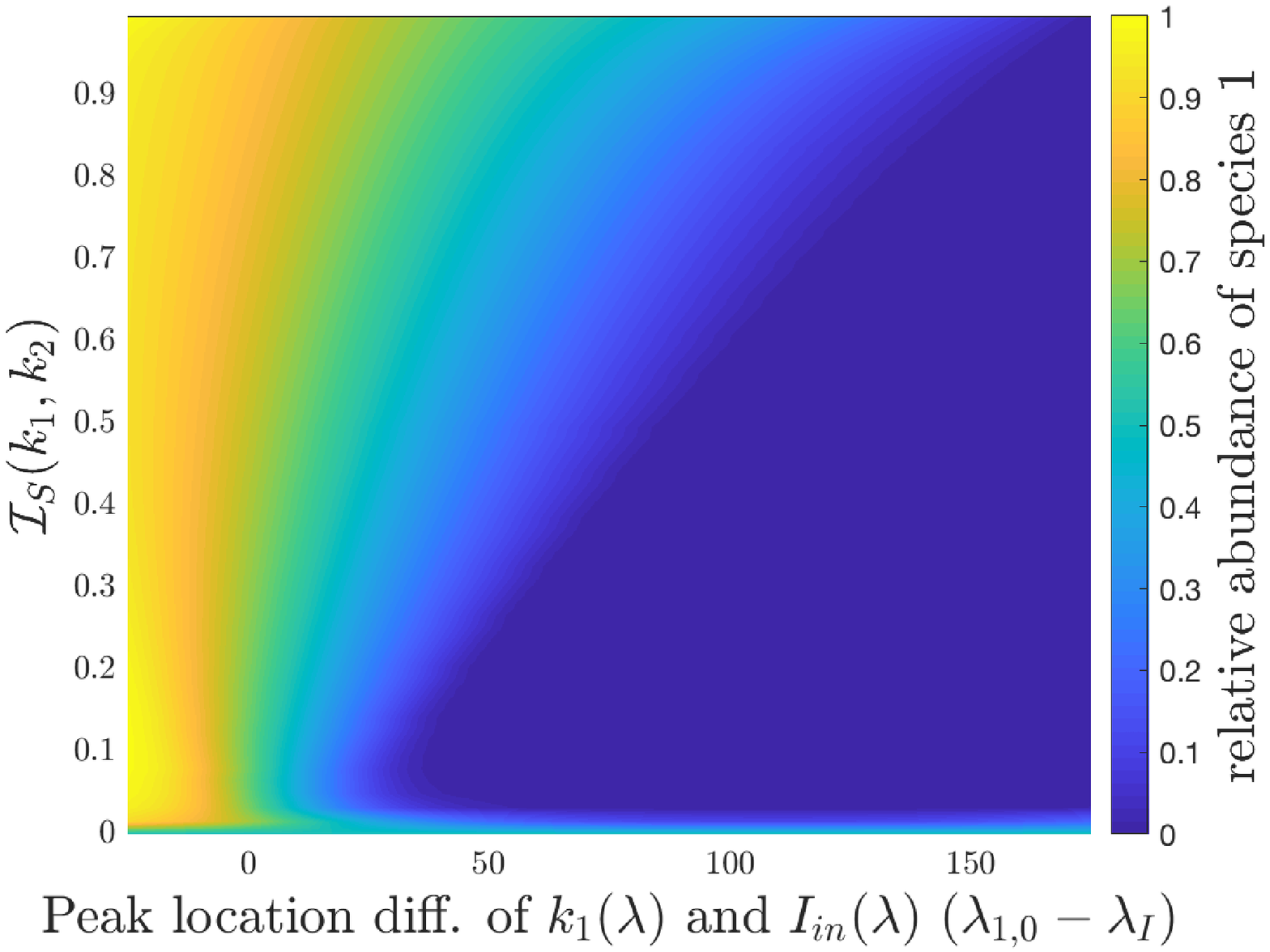}
    \caption{}
    \label{fig:specvsspecIin}
\end{subfigure}
\begin{subfigure}[b]{0.49\textwidth}
     \centering
    \includegraphics[width=\textwidth]{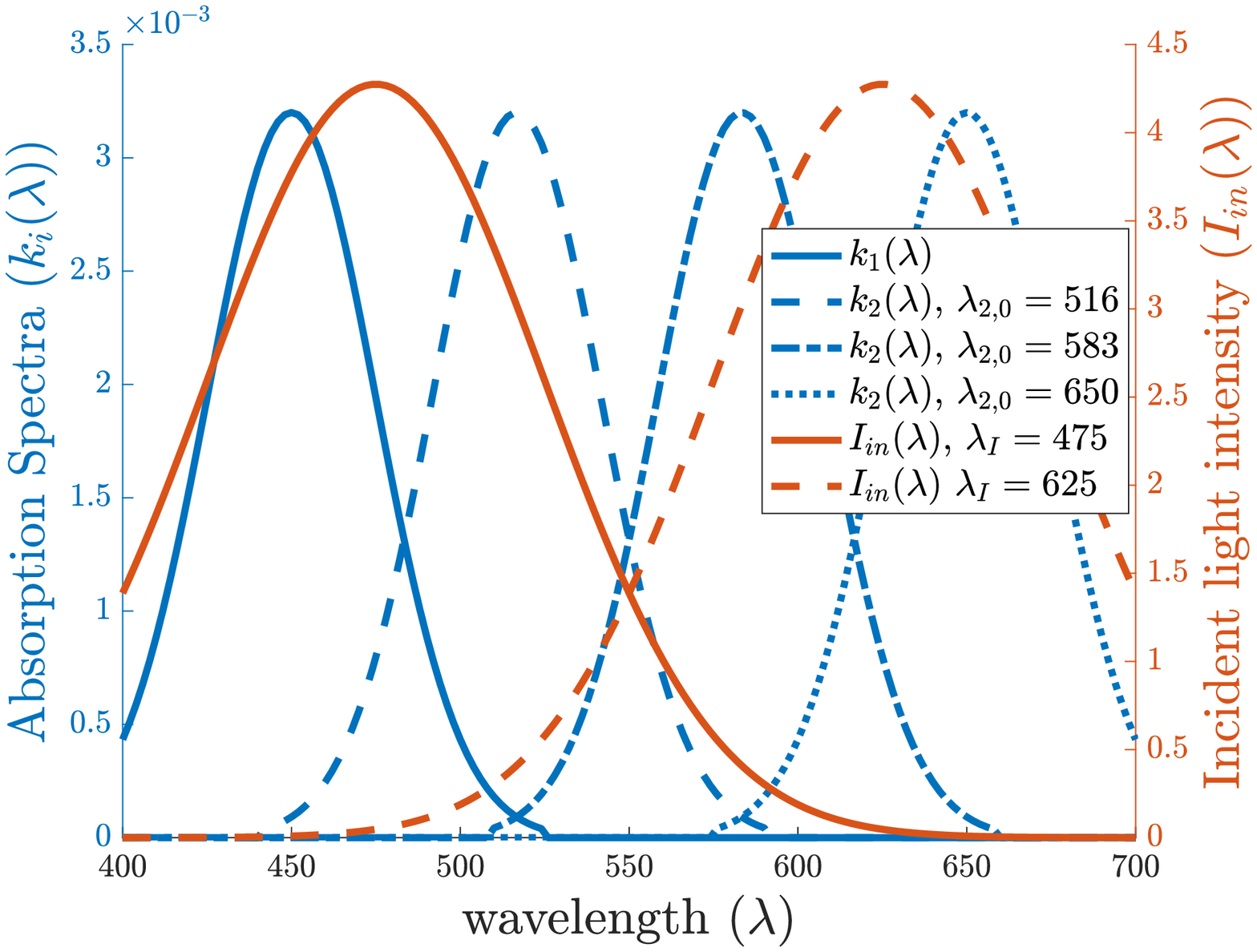}
    \caption{}
    \label{fig:explainspecvspec}
\end{subfigure}

 \caption{In (a) we show the competition outcome as the distance between the peak locations, $\lambda_{1,0} $ and $\lambda_I$, is changed versus the degree of niche differentiation between the two species (by varying $\lambda_{1,0}-\lambda_{2,0}$) The heat map is given by $\frac{|u_1|}{|u_1|+|u_2|}$. In (b) we show the shape of $k_i(\lambda)$ for four reference values of $\lambda_{i,0}$ in blue, and $I_{in}(\lambda)$ for two reference values of $\lambda_I$ in orange.}

\label{fig:specvsSpec}
\end{figure}
\begin{figure}
     \centering
    \includegraphics[width=.49\textwidth]{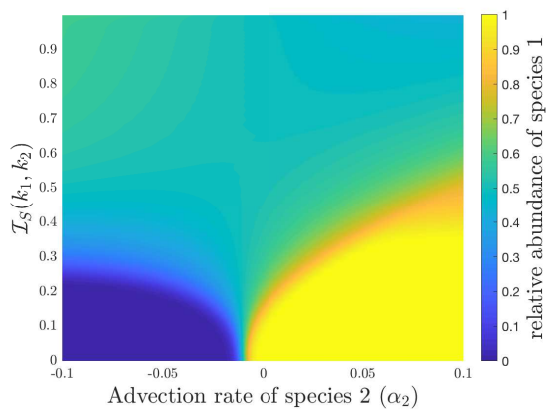}
    \caption{The competition outcome as the advection rate, $\alpha_2$, is changed versus the degree of niche differentiation between the two species under uniform incident light is shown. Example niches are given in blue in Figure \ref{fig:explainspecvspec}. The heat map is given by $\frac{|u_1|}{|u_1|+|u_2|}$. We fix $\alpha_1=-0.01$ mh$^{-1}$ and $I_{in}(\lambda)\equiv 1.67 \mu$mol$/(\text{m}^2
\cdot \text{s}\cdot \text{nm})$.}
    \label{fig:specvsspecAdvdist}
\end{figure}

Figure \ref{fig:specvsspecIin} shows the coexistence regions when varying the location of the peak of incident light and the distance between the two absorption spectra $k_1(\lambda)$ and $k_2(\lambda)$ (as measured by $\mathcal{I}_S(k_1,k_2)$). The point P1 is justified by the following observations in Figure \ref{fig:specvsspecIin}.
We see that exclusion is exhibited for extreme values of $\lambda_{1,0}-\lambda_I$ and non-zero $\mathcal{I}_S(k_1,k_2)$. When the values of $\lambda_{1,0}-\lambda_I$ are extreme, one of the species' absorption spectrum overlaps with the incident light significantly more giving it a competitive advantage. However, when the values of $\lambda_{1,0}-\lambda_I$ are intermediate and $\mathcal{I}_S(k_1,k_2)$ is large then each species has sufficient overlap with the incident light spectrum and any competitive advantage is diminished, promoting coexistence. 



 Figure \ref{fig:specvsspecAdvdist} shows the coexistence region when varying the advection rate of species 2 ($\alpha_2$) and the distance between the two absorption spectra $k_1(\lambda)$ and $k_2(\lambda)$ (given by $\mathcal{I}_S(k_1,k_2)$). 
 First, we observe that when $\mathcal{I}_S(k_1,k_2)=0$, whichever species that is more buoyant excludes the other species, as was established in Section \ref{sec:sameniche}.
 %
 However, when $\mathcal{I}_S(k_1,k_2)$ is large, the competitive exclusion caused by advection advantage is mitigated and coexistence occurs, thus justifying P2. When considering two species with unimodal absorption spectra, it is possible to overcome competitive exclusion by allowing for niche differentiation in the light spectrum.

\subsection{Outcomes for generalist versus specialist competition }
In this section we numerically explore niche differentiation in the light spectrum through competition between a specialist and a generalist. We say that a generalist species is a species whose absorption spectrum is uniform (or nearly uniform) across all visible wavelengths. Whereas we say a specialist species is one whose absorption spectrum is unimodal or narrow. In other words, a specialist absorbs a specific wavelength, or a small subset of wavelengths with a higher rate than other wavelengths. 

We explore the mechanism of specialist vs. generalist competition in overcoming competitive exclusion by explicitly comparing absorption spectra. We take $k_2(\lambda)$ to be constant (generalist) and choose $k_1(\lambda)$ such that $|k_1(\lambda)|=|k_2(\lambda)|$ in the $L^1$ norm. We further assume that $k_1(\lambda)$ is given by a truncated normal distribution, between 400 and 700 nm. By using the truncated normal distribution for $k_1(\lambda)$ we are able to change the degree of specialization of species 1 by changing the variance,$\sigma$, of the distribution as shown in Figure \ref{fig:explainspecvsgen}. Furthermore we allow the location of peak absorption to vary along the incident light spectra, that is $\lambda_{1,0}\in[400,700]$, where $\lambda_{1,0}$ is the mean of the truncated normal distribution and is the location of the local maximum of $k_1(\lambda)$. 

 We consider two scenarios to analyze the promotion of coexistence via the niche differentiation mechanism of specialist versus generalist competition. First, we assume an unimodal incident light $I_{in}(\lambda)$ as in Figure \ref{fig:specvsspecIin} and vary the location of the peak species absorption spectra, $\lambda_{1,0}$.
 Additionally, we vary the degree of specialization of species $1$ by changing the variance of the truncated normal distribution that defines its absorption spectrum. That is, by changing the variance we change the narrowness of its niche and thus change the values of $\mathcal{I}_S(k_1,k_2)$. 
 
 Second, we change $\mathcal{I}_S(k_1,k_2)$ as described above but with a uniform incident light function. We allow a competitive advantage through advection by altering the advection rate of species $2$, $\alpha_2$. Recall that $u_1$ has competitive advantage when $\alpha_1<\alpha_2$, and $u_2$ has competitive advantage when $\alpha_1>\alpha_2$ (see Theorem \ref{thm:4.3}).

 By varying $\mathcal{I}_S(k_1,k_2)$ we are able to show the competitive outcomes when niche differentiation via a specialist versus generalist competition is permitted. The results pertaining to competition outcomes of the scenarios discussed in this section are shown in Figures \ref{fig:2paramSpecivsGEn} and \ref{fig:specvsgenadv}.

\begin{figure}
     \centering
     \begin{subfigure}[b]{0.49\textwidth}
         \centering
         \includegraphics[width=1\textwidth]{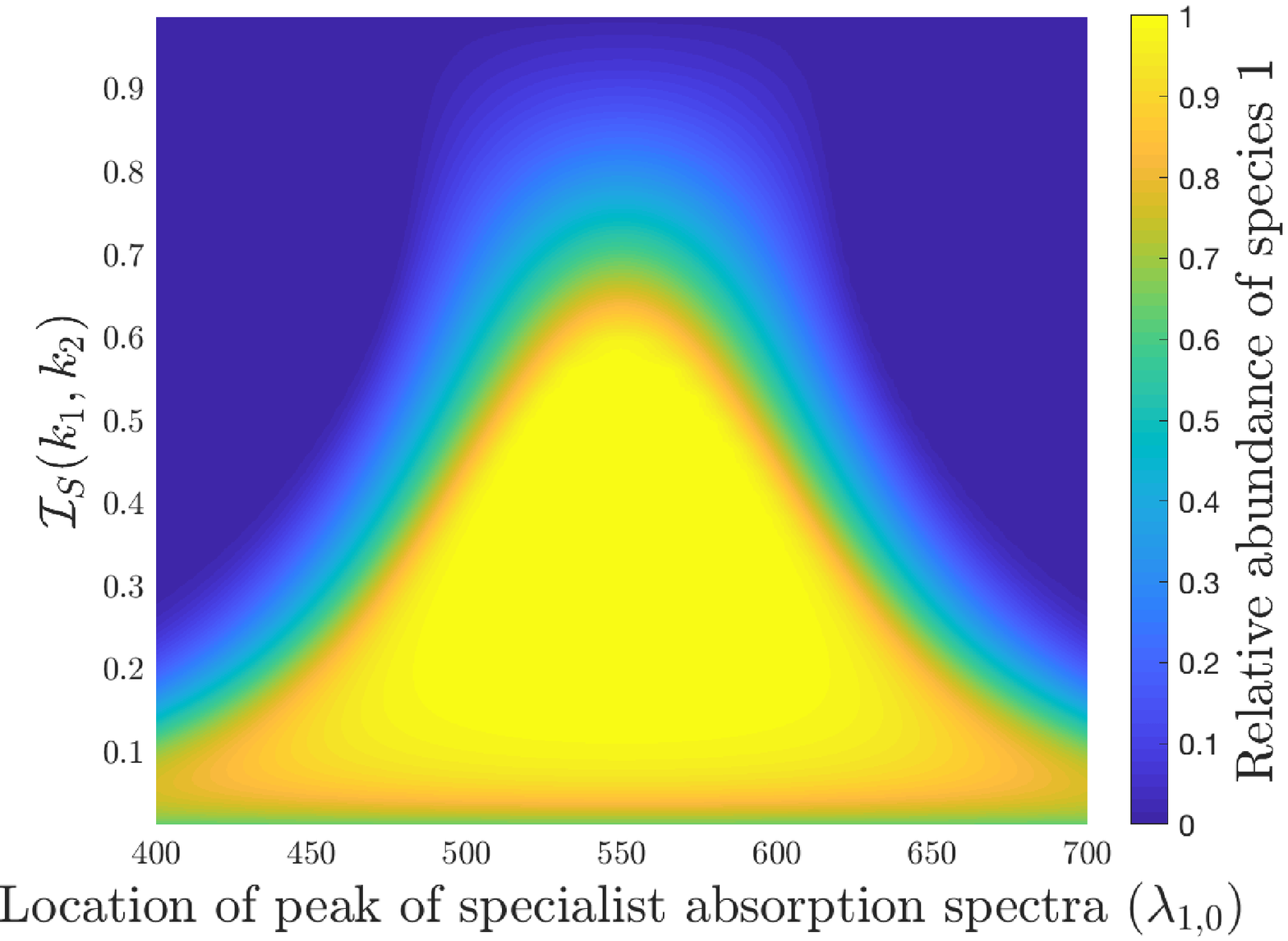}
         \caption{}
         \label{fig:specvsgenIin}
     \end{subfigure}
     \hfill
      \begin{subfigure}[b]{0.49\textwidth}
         \centering
         \includegraphics[width=1\textwidth]{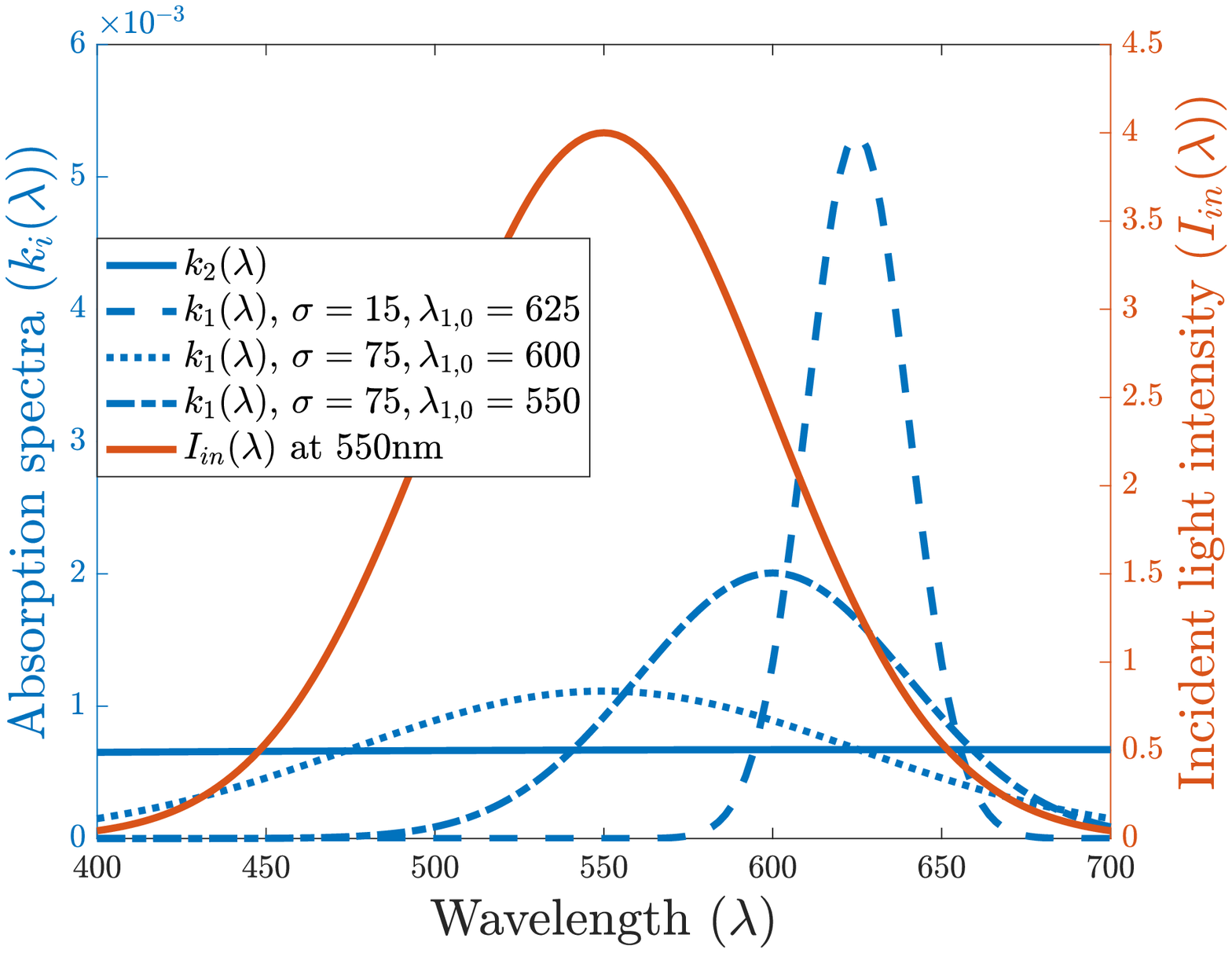}
         \caption{}
         \label{fig:explainspecvsgen}
     \end{subfigure}
        \caption{(a) shows coexistence regions for a specialist (species 1), and a generalist (species 2). The heat maps are given by $\frac{|u_1|}{|u_1|+|u_2|}$. In (b) we show samples of the absorption spectra in blue and the incident light in orange. We fix the absorption spectrum $k_2(\lambda$) of the generalist and change the specialization of species 1 by adjusting the variance of its absorption spectrum $k_1(\lambda)$ as shown by the blue lines in (b). We fix $I_{in}(\lambda)$ as given in (b) and show the competition outcome with relation to the distance between the specialists location of peak absorption and the incident lights location of peak intensity and the distance between the two absorption spectra given by $\mathcal{I}_S(k_1,k_2)$ in (a).}
        \label{fig:2paramSpecivsGEn}
\end{figure}

\begin{figure}
         \centering
         \includegraphics[width=.49\textwidth]{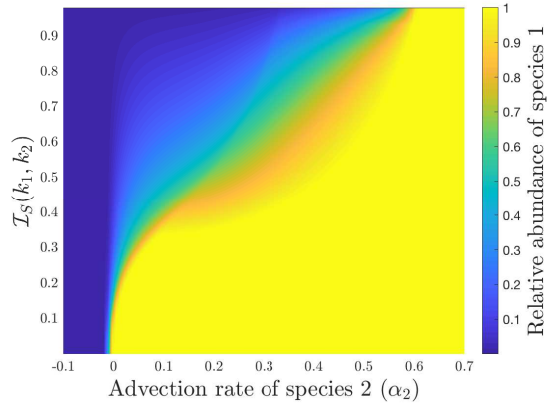}
         \caption{We show coexistence regions for competing specialist (species 1), and generalist (species 2). The heat map is given by $\frac{|u_1|}{|u_1|+|u_2|}$. We fix the absorption spectrum $k_2(\lambda$) of the generalist and adjust $\mathcal{I}_S(k_1,k_2)$ by changing the variance of $k_1(\lambda)$ while fixing the mean at $625$nm. We also vary the generalists' advection rate $\alpha_2$ from $-0.1$ to $0.7$mh$^{-1}$, while fixing the specialists' advection rate $\alpha_1=-0.01$  mh$^{-1}$.}
         \label{fig:specvsgenadv}
     \end{figure}

In Figure \ref{fig:specvsgenIin} we show the relative abundance of species 1 for various degrees of specialization and overlap with the incident light function. The point P3 is discussed in the following. Species 1 is a strong competitor for a narrow set of wavelengths, whereas species 2 is a weak competitor for a broad set of wavelengths. When species 1 is highly specialized ($\mathcal{I}_S(k_1,k_2)$ close to one) it strongly out-competes species 2 for a small portion of the light spectrum, however species 2 has little competition for the rest of the light spectrum and is able to exclude species 1. Furthermore, if species 1's niche does not overlap significantly with the incident light spectra or is too specialized then species 1 is faced with limited resource and is thus excluded. On the other hand, for intermediate specialization and relatively small distance between the location of peaks of $k_1(\lambda)$ and $I_{in}(\lambda)$ species 1 will out-compete species 2 for nearly all of the resources and thus excluding species 2. Coexistence is then observed when the specialist is relaxed to a generalist niche ($\mathcal{I}_S(k_1,k_2)$ is close to zero) because weak competition occurs along the entire light spectrum and no significant advantage is given. Additionally, for intermediate values of $\mathcal{I}_S(k_1,k_2)$ and sufficient overlap between incident light and species 1 niche, coexistence is permitted by the balance between the specialist strongly competing for a sufficient but narrow amount of resource and the generalist weakly competing for wide amount of resource that is not utilized by the specialist. 

In Figure \ref{fig:specvsgenadv} we show the relative abundance of species 1 for various degrees of specialization and advection rates of species 2 under uniform incident light. The point P2 is reiterated by the following results. Recall that in Theorem \ref{thm:4.3} we show that competitive exclusion occurs if one species has an advection advantage and there is no niche differentiation. Here we see that niche differentiation in the light spectrum ($\mathcal{I}_S(k_1,k_2)>0)$ allows for coexistence even though one species has a competitive advantage through advection. We note that if the generalist has an advection advantage then it will always exclude the specialist. On the other hand, if the specialist has the advection advantage it will exclude the generalist unless it becomes too specialized, in which case sufficient light is available for the generalist and either coexistence occurs, or in the case of extreme specialization, the specialist is excluded. Furthermore, there is a region where the competitive advantage of advection is so strong for the specialist that it will always exclude the generalist.         

 \section{Coexistence of N species}\label{sec:Nspecies}

In this section, we will show the possibility of coexistence of $N$ species, for any number $N\geq 1$. We numerically verify this result by considering competition among 5 species with varying advection rates. 
We introduce the $N$-species model analogous to  \eqref{eq:full}:

\begin{equation}\label{eq:fullN}
\begin{cases}
\partial_t u_i = D_i \partial^2_x u_i - \alpha_i \partial_x u_i + [g_i(\gamma_i(x,t)) - d_i(x)]u_i & \text{ for } 0 < x < L,~1\leq i\leq N,\\
D_i\partial_x u_i(x,t) - \alpha_i u_i(x,t) = 0 & \text{ for }x = 0,~L,\,t>0,~1\leq i\leq N,\\
u_i(x,0) = u_{i,0}(x)&\text{ for }0 < x < L,~1\leq i\leq N,
\end{cases}
\end{equation}
where $D_i>0$, $\alpha_i \in \mathbb{R}$ and $d_i$ are the diffusion rate,  buoyancy coefficient and death rate of the $i$-th species, respectively. The functions $g_i$ satisfies \eqref{eq:ggg}. 
The functions 
$\gamma_{i}(x,t)$ is the number of absorbed photons available for photosynthesis by the $i$-th species and is given by
\begin{equation}\label{eq:gammai}
\gamma_i(x,t) = \int_{400}^{700}k_i(\lambda) I(\lambda,x)\,d\lambda,
\end{equation}
where we have chosen $a_i \equiv 1$ as before, and
\begin{equation}\label{eq:II}
  I(\lambda,x) = I_{\rm in}(\lambda) \exp\left[- K_{BG}(\lambda)x - \sum_{i=1}^N k_i(\lambda)\int_0^x u_i(y,t)\,dy \right].
\end{equation}

\begin{theorem}
    Let the incident light spectrum $I_{in}(\lambda)$ be positive on an open set in $[400,700]$. Then for each $N \geq 1$, there exists a choice of $d_i$ and $\{k_i(\lambda)\}_{i=1}^N$ such that all $N$ species can persist in \eqref{eq:fullN}, i.e. for any positive initial condition, the solution $(u_i)_{i=1}^N$ of \eqref{eq:fullN} satisfies
    $$
    \liminf_{t\to\infty} \left[\inf_{0\leq x \leq L} u_i(x,t) \right]>0 \quad \text{ for each }1\leq i \leq N.
    $$
\end{theorem}
\begin{proof}
By the hypotheses of the theorem, there exists $\lambda_1,\lambda_2$ such that  $400 \leq \lambda_1 < \lambda_2 \leq 700$ and that $I_*:= \inf_{[\lambda_1,\lambda_2]}I_{in}(\lambda) >0$. Let $\{J_i\}_{i=1}^N$ be a partition of $[\lambda_1,\lambda_2]$, and choose the functions $k_i(\lambda)$ such that $\textup{Supp}\,k_i \subset {\rm Int}\, J_i$. In particular, the support of $k_i$ do not overlap. Hence, 
$$
  I(\lambda,x) = I_{\rm in}(\lambda) \exp\left[- K_{BG}(\lambda)x - k_i(\lambda)\int_0^x u_i(y,t)\,dy\right] \quad \text{ in }\textup{Supp}\,k_i, 
$$
and the $i$-th species satisfies effectively a single species equation
$$
\begin{cases}
\partial_t u_i = D_i \partial^2_xu_i - \alpha_i \partial_x u_i+ [g_i(\gamma_i(x,t)) - d_i(x)]u_i &  \text{ for } 0 < x < L,\, t>0,\\
D_i \partial_x u_i - \alpha_i \partial_x u_i = 0 &\text{ for }x = 0,L,\, t>0,
\end{cases}
$$
with $\gamma_i$ being independent of $u_j$ for $j\neq i$. Precisely,
\begin{equation}\label{eq:gammaii}
\gamma_i(x,t) = \int_{400}^{700} k_i(\lambda) I_{\rm in}(\lambda) \exp\left[- K_{BG}(\lambda)x - k_i(\lambda)\int_0^x u_i(y,t)\,dy\right]\,d\lambda.
\end{equation}

Next, we choose $d_i$ to be a positive constant such that
$$
\mu(D_i,\alpha_i,g_i(\int k_i(\lambda)I_{in}(\lambda) \exp(-K_{BG}(\lambda)x)) - d_i) <0.
$$
This is possible since
\begin{align*}
    \lim_{d_i \to 0}\mu(D_i,\alpha_i,&g_i(\int k_i(\lambda)I_{in}(\lambda) \exp(-K_{BG}(\lambda)x)) - d_i)
    \\&=\mu(D_i,\alpha_i,g_i(\int k_i(\lambda)I_{in}(\lambda) \exp(-K_{BG}(\lambda)x)))<0,
\end{align*}

where the last inequality follow from Lemma \ref{lem:eigen}. 
It then follows from~\cite[Proposition 3.11]{Jiang2019} that 
the problem 
$$
\begin{cases}
D_i \partial^2_x u_i - \alpha_i \partial_x u_i + [g_i(\tilde\gamma_i(x)) - d_i]u_i & \text{ for } 0 < x < L,\\
D_i\partial_x u_i(x) - \alpha_i u_i(x) = 0 & \text{ for }x = 0,~L,
\end{cases}
$$
with $\tilde\gamma_i(x)$ given by 
$$
\hat\gamma_i(x)=\int_{400}^{700}a_i(\lambda) k_i(\lambda) I_{\rm in}(\lambda) \exp\left[- K_{BG}(\lambda)x - k_i(\lambda)\int_0^x \tilde{u}_i(y,t)\,dy\right]\,d\lambda,
$$
has a unique positive solution $\tilde{u}_i$. Moreover, 
$$
u_i(\cdot,t) \to \tilde{u}_i \quad \text{ in }C([0,L]),\text{ as }t \to \infty,
$$
provided $u_i(\cdot,0) \not\equiv 0.$  
This completes the proof.
\end{proof}

Next, we numerically demonstrate the possibility of coexistence of five phytoplankton species under niche differentiation. We assume that all five species $D_i=D_j$, $d_i=d_j$, and $g_i=g_j$ for all $i,j$ and that $\alpha_1=0.01$ with $\alpha_i=i\cdot\alpha_1$ for all $i$. We assume that all absorption spectra are unimodal and are given by the truncated normal distribution. Furthermore, each absorption spectrum $k_i(\lambda)$ is a horizontal translation of one another. We alter the location of peak absorption (or the mean) to allow for niche differentiation similarly to Figures \ref{fig:explainspecvspec} and \ref{fig:explainspecvsgen}. We also assume that the incident light ($I_{in}(\lambda)$) is unimodal with peak absorption located at 575 nm allowing for a competitive advantage.  
We compare the relative abundances of the five species at time $t$ defined by
\begin{equation}
    \bar u_i(t)=\frac{|u_i(x,t)|_{L^1}}{\sum_{j=1}^N |u_j(x,t)|_{L^1}},
\end{equation}
where the $L^1$ norm here is taken with respect to the spatial variable $x$. We further denote the relative abundance at equilibrium as $\bar u_i^*$. 
In addition, we define the $N$ species niche differentiation index as
\begin{equation}\label{eq:overlapmeasure}
    \mathcal{I}_i=\frac{1}{N-1}\sum_{j=1,j\neq i}^N \mathcal{I}_S(k_i,k_j).
\end{equation}
 Consequently the average niche differentiation index is given as 
\begin{equation}\label{eq:Ibar}
 \bar{\mathcal{I}}=\frac{1}{N}\sum_i^N \mathcal{I}_i.
\end{equation}

\begin{figure}
     \centering
         \includegraphics[width=0.8\textwidth]{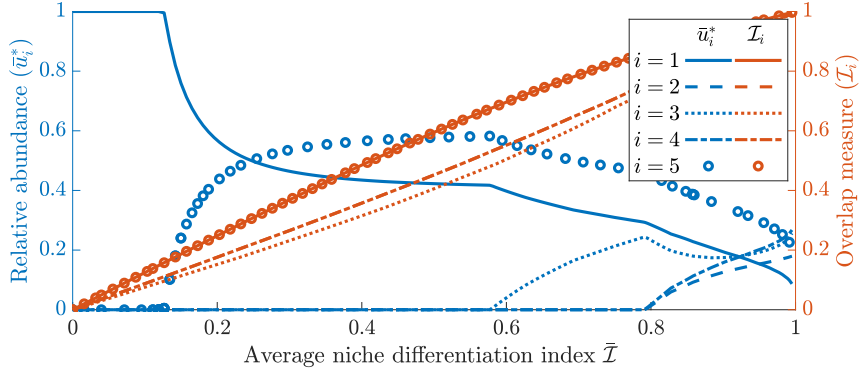}
        \caption{Gives the steady state relative abundance ($\bar u_i^*$) of 5 competing species and their respective overlap measure defined in \eqref{eq:overlapmeasure}. The x-axis is labelled as the average overlap measure $\bar{\mathcal{I}}$ given in \eqref{eq:Ibar}. All other model parameters are the same among species expect the competitive advantage obtained through buoyancy: $\alpha_1=0.01$ mh$^{-1}$, $\alpha_2=0.02$ mh$^{-1}$, $\alpha_3=0.03$ mh$^{-1}$, $\alpha_4=0.04$ mh$^{-1}$, $\alpha_5=0.05$ mh$^{-1}$.} 
        \label{fig:5SP} 
\end{figure}

Figure \ref{fig:5SP} gives the numerical results of the five species competition. Competitive exclusion occurs when niche differentiation is not sufficient and the species with the lowest advection rate (species 1) excludes all other species. However, as the niche differentiation is increased, more species are able to coexist and all five species can persist when niche differentiation is significant enough. 

\section{Red versus Green cyanobacteria competition}

In this section we numerically explore a more realistic competition scenario between two phytoplankton species. To incorporate realistic biological assumptions into our model we consider two main things. First, the background attenuation of water is not uniform across the visible light spectrum and depends on the amount of dissolved and particulate organic matter (gilvin and tripton) in the water.  Second, the absorption spectra considered in Section \ref{sec:NumCoexistmechs} are idealistic for investigation and are not typical for a phytoplankton species. Thus, in this section we consider absorption spectra given empirically as in Figure \ref{fig:absspectra} and explore competition outcomes. 

\subsection{Background attenuation in water}
Here we introduce a reasonable function to more accurately model background attenuation of water, gilvin and tripton and phytoplankton.

We divide the background attenuation into two parts to account for the attenuation of pure water and gilvin and tripton 
\begin{equation}
    K_{BG}(\lambda)=K_W(\lambda) +K_{GT}(\lambda),
\end{equation}
where $K_W(\lambda)$ is readily found in the literature and shown in Figure \ref{fig:waterabsorb}~\cite{Stomp2007,Pope1997}. $K_{GT}(\lambda)$ is also found in literature and is given by the following form \cite{kirk2010}:
\begin{equation}
    K_{GT}(\lambda)=K_{GT}(\lambda_r)\textup{exp}(-S(\lambda-\lambda_r)),
\end{equation}
where $\lambda_r$ is a reference wavelength with a known turbidity and $S$ is the slope of the exponential decline. Following literature we take reasonable values for each of these variables with $S=0.017 $nm$^{-1}$ as in~\cite{Stomp2007} and referenced in~\cite{kirk2010}. We fix our reference wavelength, $\lambda_r$, to be $480$nm.  
The background attenuation is larger in turbid lakes due to the high concentrations of gilvin and tripton. For this reason, we use $K_{GT}(480)$ as a proxy for the turbidity of a lake, and vary $K_{GT}(480)$ between $0.1-3\textup{m}^{-1}$. That is, low $K_{GT}(480)$ values correspond to clear lakes whereas high $K_{GT}(480)$ values correspond to  highly turbid lakes.

\begin{figure}
    \centering
    \includegraphics[width=0.6\paperwidth]{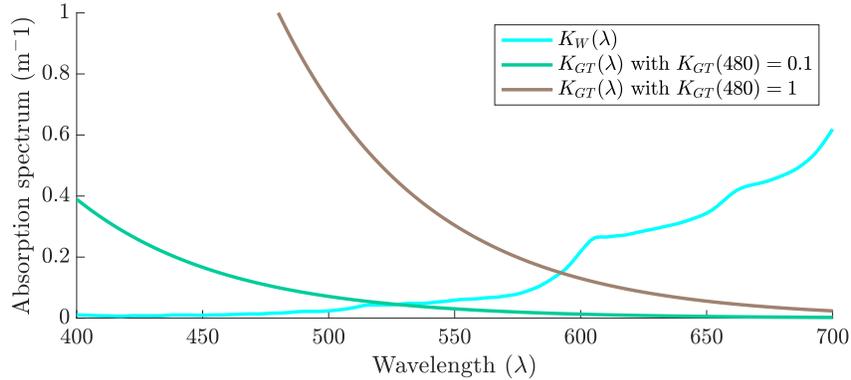}
    \caption{The absorption spectrum of pure water~\cite{Pope1997,Stomp2007}, and the absorption spectra for lakes with gilvin and tripton concentrations representative of clear oligotrophic or mesotrophic waters ($K_{BG}(480)=0.1$), and turbid eutrophic waters ($K_{BG}(480)=1$).}
    \label{fig:waterabsorb}
\end{figure}
Lastly, we consider the absorption spectra of red and green cyanobacteria species. In Figure \ref{fig:absspectra} we see that there are significant differences in the absorption spectra between the phytoplankton allowing for niche differentiation. 

\subsection{Competition outcomes of red and green cyanobacteria}
We now show the steady state outcome when red and green cyanobacteria compete for light in lakes of varying turbidity. 

\begin{figure}
    \centering
    \includegraphics[width=0.7\paperwidth]{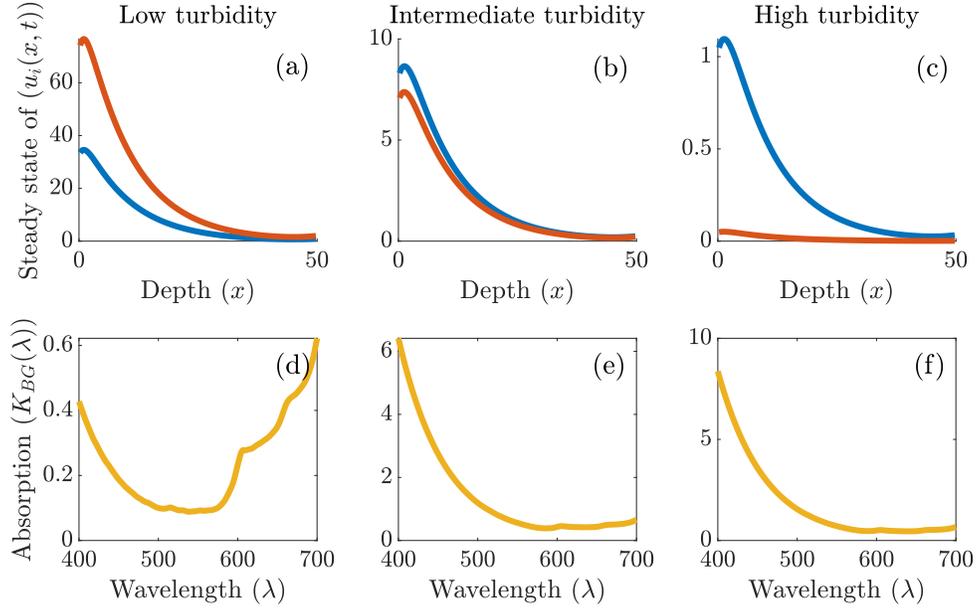}
    \caption{(a)-(c) show steady state outcomes of competition between green cyanobacteria, $u_1(x,t)$ (shown in blue), and red cyanobacteria, $u_2(x,t)$ (shown in red), for various amounts of gilvin and tripton that correspond to low, intermediate and high turbidity, respectively. (d)-(f) shows the background absorption for those states with $K_{BG}(480)=0.1$, $K_{BG}(480)=1.1$, $K_{BG}(480)=2$, respectively.} 
    \label{fig:reallife}
\end{figure}

In Figure \ref{fig:reallife} the competition outcome between green cyanobacteria (\textit{Synechocystis} strain) and red cyanobacteria (\textit{Synechococcus} strain) is shown. In Figure \ref{fig:absspectra}, the green cyanobacteria absorption spectra is shown in blue and the red cyanobacteria is shown in red. Their absorption spectra are sufficiently different so that niche differentiation occurs. That is, the green cyanobacteria mainly absorbs light in the orange-red ranges, whereas the red cyanobacteria absorbs more green light. Both species absorb blue light similarly. Thus, the light availability throughout the water column plays an important role in competition outcome. In Figures \ref{fig:reallife}(d)-(f) we see that as the gilvin and tripton concentrations increase (shifting from low turbidity to high) the background absorption's shift to absorb proportionally more blue and green light, leaving proportionally more orange and red light available. This shift in available light then modifies the competitive outcome, where red cyanobacteria clearly dominate in less turbid case, whereas green cyanobacteria dominate in the highly turbid case, even though the two species coexist in both situations.

\section{Conclusion}

In this manuscript we explore niche differentiation along the light spectrum by extending the models of Stomp et al. \cite{Stomp2007} to the spatial context, using well established reaction-diffusion approach. Differing with previous works~\cite{Jiang2019,HsuLou2010,Du2011}, in which light was regarded as a single resource with varying intensity, here we treat light as a continuum of resources that have varying availability and are consumed in different efficiency by the phytoplankton species.
Our main theoretical results, found in Section \ref{sec:mathresults}, stem from the theory of monotone dynamical systems and include the existence and attractiveness of the equilibrium. These results give a condition for when the semi-trivial equilibria exist and characterize their stability. As an extension, a condition for coexistence is obtained. The condition for coexistence is then made explicit to offer direct biological interpretations based on model parameters.  Niche differentiation is introduced in Section \ref{sec:Mathniche} by allowing the absorption spectra ($k_i(\lambda)$) of competing species to change. We consider the case where the competing species niches are completely disjoint and provide a condition for coexistence. Furthermore we consider the case when competing species occupy the same niche and provide competitive exclusion outcomes based on transport related parameters and show that species who are able to stay closer to the surface through either advection or turbulent diffusion will competitively exclude. These results lay the groundwork to study the impacts niche differentiation will have on coexistence outcomes in Section \ref{sec:NumCoexistmechs}. 

We show numerically, in Section \ref{sec:NumCoexistmechs}, a myriad of mechanisms in which coexistence can occur. When two specialists compete, the competitive advantages given by advection or incident light can be overcome when niche differentiation is significant. This is shown in Figures \ref{fig:specvsSpec}. Furthermore, we see that competitive exclusion occurs when the overlap between the incident light and a species' absorption spectrum is large, see Figure \ref{fig:specvsspecIin}. In addition, the more buoyant species no longer dominates if niche differentiation is significant, as shown in Figure \ref{fig:specvsspecAdvdist}.   Similarly, in the competition between a specialist and a generalist,  coexistence readily occurs for intermediate degrees of niche differentiation. However,if the niche of the specialist occupies only a narrow part of the incident light spectrum, then 
their growth rate can be negatively impacted as shown in Figure \ref{fig:2paramSpecivsGEn}. In either case niche differentiation in the light spectrum is enough to overcome competitive exclusion caused by diffusion and advection, thus offering an important perspective in resolving the paradox of the plankton in the affirmative direction. 

Furthermore, to fully explore the ecological diversity and the paradox of the plankton, we consider a system with $N$ competing species. First, we show analytically that coexistence of $N$ species is possible under sufficient niche differentiation and proper natural death rate ($d_i(\lambda)$) functions. This result  suggests a possible evolutionary strategies that phytoplankton may take in partitioning  in their usage of the light spectrum for growth~\cite{Holtrop2021}. To illustrate our result, we provide numerical simulations for a five species competition scenario with an advection and incident light advantage present. Here we choose phytoplankton species that have differential buoyancy properties. In the absence of niche differentiation, competitive exclusion were predicted by previous work \cite{Jiang2019}. When niche differentiation is significant, we observe that the species are able to coexist in a robust manner.

Lastly, we numerically study the competition dynamics for absorption spectra and background attenuation functions that are representative of phytoplankton species found in nature. Precisely, we consider the absorption spectra of green and red cyanobacteria species and explore the competitive outcome as it depends on the nutrient status, or turbidity of the ecosystem as shown in Figure \ref{fig:reallife}. Our numerical results suggest that clear lakes host higher abundances of red cyanobacteria whereas green cyanobacteria out-compete in highly turbid, eutrophic lakes. Our result in particular confirms with empirical results~\cite{Stomp2007} and is potentially useful in understanding phytoplankton competition.

In this paper, we explored a potential explanation to the paradox of the plankton by allowing for niche differentiation in the visible light spectrum. To achieve this, we made several simplifying assumptions about the biological system, such as our sufficient nutrient assumption. It is well known that phytoplankton dynamics heavily depend on nutrient dynamics~\cite{Whitton2012,Klausmeier2004b,Reynolds2006}. Thus, in order to fully understand phytoplankton population dynamics, future attempts at modelling niche differentiation should also allow for the explicit consideration of nutrient and nutrient uptake dynamics. We have also assumed that our model parameters are constant in time. This in general is not true for ecological systems, and in particular those that explicitly consider light. Light availability is periodic on the time scales of days and, in addition, periodic seasonally. In addition to light, parameters related to mortality and motility can depend on water temperature and thus change seasonally. This type of oscillatory forcing can significantly change dynamics and especially when considering transient dynamics~\cite{Hastings2018}. 

Even though our model can be improved in various ways, our results are biologically intuitive and are consistent with the current state of the biological literature. Our work furthers the understanding of niche differentiation and phytoplankton competition and can be used as a basis for future studies of phytoplankton dynamics and predictive modelling. In conclusion, our study shows that niche differentiation can promote coexistence of phytoplankton species in a robust way, thus supporting one explanation of the Hutchinson's paradox.


\bibliographystyle{ieeetr}
\bibliography{chrisbib}

\begin{thebibliography}{10}

\bibitem{Huisman2018}
J.~Huisman, G.~A. Codd, H.~W. Paerl, B.~W. Ibelings, J.~M. Verspagen, and P.~M.
  Visser, ``{Cyanobacterial blooms},'' 6 2018.

\bibitem{Reynolds2006}
C.~S. Reynolds, {\em {The ecology of phytoplankton}}.
\newblock Cambridge University Press, 1 2006.

\bibitem{Watson2015}
S.~B. Watson, B.~A. Whitton, S.~N. Higgins, H.~W. Paerl, B.~W. Brooks, and
  J.~D. Wehr, ``{Harmful Algal Blooms},'' in {\em Freshwater Algae of North
  America: Ecology and Classification}, no.~March 2017, pp.~873--920, Elsevier
  Inc., 2015.

\bibitem{Paerl2013}
H.~W. Paerl and T.~G. Otten, ``{Harmful Cyanobacterial Blooms: Causes,
  Consequences, and Controls},'' {\em Microbial Ecology}, vol.~65, no.~4,
  pp.~995--1010, 2013.

\bibitem{Hutchinson1961}
G.~E. Hutchinson, ``{The Paradox of the Plankton},'' {\em The American
  Naturalist}, vol.~95, pp.~137--145, 10 1961.

\bibitem{Jiang2019}
D.~Jiang, K.~Y. Lam, Y.~Lou, and Z.~C. Wang, ``{Monotonicity and global
  dynamics of a nonlocal two-species phytoplankton model},'' {\em SIAM Journal
  on Applied Mathematics}, vol.~79, pp.~716--742, 4 2019.

\bibitem{Jiang2021}
D.~Jiang, K.~Y. Lam, and Y.~Lou, ``{Competitive exclusion in a nonlocal
  reaction–diffusion–advection model of phytoplankton populations},'' {\em
  Nonlinear Analysis: Real World Applications}, vol.~61, p.~103350, 10 2021.

\bibitem{HsuLou2010}
S.~B. Hsu and Y.~Lou, ``{Single phytoplankton species growth with light and
  advection in a water column},'' {\em SIAM Journal on Applied Mathematics},
  vol.~70, no.~8, pp.~2942--2974, 2010.

\bibitem{Heggerud2020}
C.~M. Heggerud, H.~Wang, and M.~A. Lewis, ``{Transient dynamics of a
  stoichiometric cyanobacteria model via multiple-scale analysis},'' {\em SIAM
  Journal on Applied Mathematics}, vol.~80, no.~3, pp.~1223--1246, 2020.

\bibitem{Huisman1994}
J.~Huisman and F.~J. Weissing, ``{Light limited growth and competition for
  light in well mixed aquatic environments: an elementary model},'' {\em
  Ecology}, vol.~75, no.~2, pp.~507--520, 1994.

\bibitem{Wang2007}
H.~Wang, H.~Smith, Y.~Kuang, and J.~J. Elser, ``{Dynamics of stoichiometric
  bacteria-algae interactions in the epilimnion},'' {\em SIAM Journal on
  Applied Mathematics}, vol.~68, no.~2, pp.~503--522, 2007.

\bibitem{Burson2018}
A.~Burson, M.~Stomp, E.~Greenwell, J.~Grosse, and J.~Huisman, ``{Competition
  for nutrients and light: testing advances in resource competition with a
  natural phytoplankton community},'' {\em Ecology}, vol.~99, pp.~1108--1118, 5
  2018.

\bibitem{Luimstra2020}
V.~M. Luimstra, J.~M. Verspagen, T.~Xu, J.~M. Schuurmans, and J.~Huisman,
  ``{Changes in water color shift competition between phytoplankton species
  with contrasting light-harvesting strategies},'' {\em Ecology}, vol.~101,
  p.~e02951, 3 2020.

\bibitem{Holtrop2021}
T.~Holtrop, J.~Huisman, M.~Stomp, L.~Biersteker, J.~Aerts, T.~Gr{\'{e}}bert,
  F.~Partensky, L.~Garczarek, and H.~J. v.~d. Woerd, ``{Vibrational modes of
  water predict spectral niches for photosynthesis in lakes and oceans},'' {\em
  Nature Ecology and Evolution}, vol.~5, pp.~55--66, 1 2021.

\bibitem{Stomp2007a}
M.~Stomp, J.~Huisman, L.~J. Stal, and H.~C. Matthijs, ``{Colorful niches of
  phototrophic microorganisms shaped by vibrations of the water molecule},'' 8
  2007.

\bibitem{Du2011}
Y.~Du and L.~Mei, ``{On a nonlocal reaction-diffusion-advection equation
  modelling phytoplankton dynamics},'' {\em Nonlinearity}, vol.~24, no.~1,
  pp.~319--349, 2011.

\bibitem{Shigesada1981}
N.~Shigesada and A.~Okubo, ``{Analysis of the self-shading effect on algal
  vertical distribution in natural waters},'' {\em Journal of Mathematical
  Biology}, vol.~12, no.~3, pp.~311--326, 1981.

\bibitem{Zhang2021}
J.~Zhang, J.~D. Kong, J.~Shi, and H.~Wang, ``{Phytoplankton Competition for
  Nutrients and Light in a Stratified Lake: A Mathematical Model Connecting
  Epilimnion and Hypolimnion},'' {\em Journal of Nonlinear Science}, vol.~31,
  pp.~1--42, 3 2021.

\bibitem{Smith}
H.~L. Smith, {\em {Monotone dynamical systems: an introduction to the theory of
  competitive and cooperative systems: an introduction to the theory of
  competitive and cooperative systems}}.
\newblock American Mathematical Soc., no. 41~ed., 2008.

\bibitem{Stomp2007}
M.~Stomp, J.~Huisman, L.~V{\"{o}}r{\"{o}}s, F.~R. Pick, M.~Laamanen,
  T.~Haverkamp, and L.~J. Stal, ``{Colourful coexistence of red and green
  picocyanobacteria in lakes and seas},'' {\em Ecology Letters}, vol.~10,
  pp.~290--298, 4 2007.

\bibitem{Du2010}
Y.~Du and S.~B. Hsu, ``{On a nonlocal reaction-diffusion problem arising from
  the modeling of phytoplankton growth},'' {\em SIAM Journal on Mathematical
  Analysis}, vol.~42, no.~3, pp.~1305--1333, 2010.

\bibitem{Hess}
P.~Hess, {\em {Periodic-parabolic boundary value problems and positivity}}.
\newblock Harlow: Longman Scientific {\&} Technical, 1991.

\bibitem{Hsu1996competitive}
S.~B. Hsu, H.~L. Smith, and P.~Waltman, ``{Competitive exclusion and
  coexistence for competitive systems on ordered Banach spaces},'' {\em
  Transactions of the American Mathematical Society}, vol.~348, no.~10,
  pp.~4083--4094, 1996.

\bibitem{Lam2016remark}
K.-Y. Lam and D.~Munther, ``{A remark on the global dynamics of competitive
  systems on ordered Banach spaces},'' {\em Proceedings of the American
  Mathematical Society}, vol.~144, pp.~1153--1159, 3 2016.

\bibitem{Pope1997}
R.~M. Pope and E.~S. Fry, ``{Absorption spectrum (380–700 nm) of pure water
  II Integrating cavity measurements},'' {\em Applied Optics}, vol.~36,
  p.~8710, 11 1997.

\bibitem{kirk2010}
J.~T. Kirk, {\em {Light and photosynthesis in aquatic ecosystems}}.
\newblock Cambridge: Cambridge University Press, third~ed., 2010.

\bibitem{Whitton2012}
B.~A. Whitton, {\em {Ecology of cyanobacteria II: Their diversity in space and
  time}}.
\newblock Springer Netherlands, 2012.

\bibitem{Klausmeier2004b}
C.~A. Klausmeier, E.~Litchman, and S.~A. Levin, ``{Phytoplankton growth and
  stoichiometry under multiple nutrient limitation},'' {\em Limnology and
  Oceanography}, vol.~49, pp.~1463--1470, 7 2004.

\bibitem{Hastings2018}
A.~Hastings, K.~C. Abbott, K.~Cuddington, T.~Francis, G.~Gellner, Y.~C. Lai,
  A.~Morozov, S.~Petrovskii, K.~Scranton, and M.~L. Zeeman, ``{Transient
  phenomena in ecology},'' {\em Science}, vol.~361, no.~6406, 2018.

\end{thebibliography}

\appendix

\setcounter{secnumdepth}{0}
\section{Appendix}
\renewcommand\thesection{\Alph{section}}
\addtocounter{section}{1}
In this appendix, we recall several useful lemmas concerning the principal eigenvalue $\mu(D,\alpha,h)$ of \eqref{eq:simpleeigenvalueproblem}.

\begin{lemma}\label{lem:eigen}

 Suppose either (i) $ \int_0^L e^{\alpha x/D} h(x) dx> 0$, or (ii) $ \int_0^L e^{\alpha x/D} h(x) dx=0$, and $h'(x)$ is not identically zero in $[0,L]$, then $\mu(D,\alpha,h)<0$.

\end{lemma}
\begin{proof}
Let $\tilde\phi(x)= e^{-\alpha x/D}\phi(x)$, where $\phi$ is a principal eigenfunction of $\mu(D,\alpha,h)$, and satisfies $\phi>0$ in $[0,L]$. Then \eqref{eq:simpleeigenvalueproblem} can be rewritten as
\begin{equation}\label{eq:expevalproblem}
\begin{cases}
0 = D \partial_{x}\left( e^{\alpha x/D} \partial_x \tilde\phi \right) + e^{\alpha x/D}(h(x) + \mu)\tilde\phi &\text{ for }(x) \in [0,L],\\
\partial_x \tilde \phi = 0&\text{ for }(x) \in \{0,L\}.
\end{cases}    
\end{equation}
Notice that $\tilde\phi >0$ in $[0,L]$, by the strong maximum principle. One can divide the above equation by $\tilde\phi$ and integrate over $[0,L]$ to get
\begin{align*}
0 &= D \int_0^L \frac{1}{\tilde\phi}\partial_{x}\left( e^{\alpha x/D} \partial_x \tilde\phi \right)\,dx + \int_0^L e^{\alpha x/D}(h(x) + \mu)\,dx,\\
&=- D \int_0^L \partial_x \left( \frac{1}{\tilde\phi}\right) \left( e^{\alpha x/D} \partial_x \tilde\phi \right)\,dx + \int_0^L e^{\alpha x/D}(h(x) + \mu)\,dx,\\
&= D \int_0^L e^{\alpha x/D}  \frac{|\partial_x \tilde\phi|^2}{\tilde\phi^2}\,dx +  \int_0^L e^{\alpha x/D}(h(x) + \mu)\,dx.
\end{align*}
Note that we used the Neumann boundary condition of $\tilde\phi$ to perform the integrate by parts in the second equality. Hence,
\begin{equation}\label{eq:eigeneq}
    -\mu \int_0^L e^{\alpha x/D}\,dx =  D  \int_0^L e^{\alpha x/D}  \frac{|\partial_x \tilde\phi|^2}{\tilde\phi^2}\,dx +  \int_0^L e^{\alpha x/D}h(x)\,dx.
\end{equation}
Suppose to the contrary that $\mu \geq 0$, then it follows from \eqref{eq:eigeneq} that $ \int_0^L e^{\alpha x/D}h(x)\,dx \leq 0.$ Hence, case (i) is impossible, and we must have case (ii), which implies
$$
\int_0^L e^{\alpha x/D}  \frac{|\partial_x \tilde\phi|^2}{\tilde\phi^2}\,dx =  \int_0^L e^{\alpha x/D}h(x)\,dx=0.
$$
Hence, $\partial_x\tilde\phi \equiv 0$ which by \eqref{eq:expevalproblem} (in the Appendix) implies either $\tilde{\phi}\equiv 0$ or $h(x)\equiv 0$,  which leads to a contradiction. 

\end{proof}

\begin{lemma}
If $h(x) \in C^1([0,L])$ satisfies $h'(x) <0$ in $[0,L]$, then 
\begin{itemize}
    \item[{\rm(a)}] $\displaystyle \frac{\partial \mu}{\partial \alpha}(D,\alpha,h) >0$ for any $D>0$ and $\alpha \in \mathbb{R}.$
    \item[{\rm(b)}] $\displaystyle \frac{\partial D}{\partial \alpha}(D,\alpha,h) >0$ for any $D>0$ and $\alpha \leq 0.$
    \item[{\rm(c)}] If $\mu(D_0,\alpha_0,h)=0$ for some $D_0$ and $\alpha_0 \geq h(0)L$, then $\frac{\partial \mu}{\partial D}(D_0,\alpha_0,h) <0.$
\end{itemize}
\label{lem:4.3}
\end{lemma}
\begin{proof}
Assertion (a) follows from~\cite[Lemma 4.8]{Jiang2019}, while assertions (b) and (c) follow from~\cite[Lemma 4.9]{Jiang2019}.
\end{proof}

\end{document}